\documentclass[12pt]{article}
\usepackage[margin=2cm]{geometry}
\usepackage{amsthm}
\usepackage{amsmath}
\usepackage{amssymb}
\usepackage{graphicx}
\usepackage{marginnote}
\usepackage{tikz-cd}
\usepackage{tikz}
\tikzset{
  treenode/.style = {shape=rectangle, rounded corners,
                     draw, align=center,
                     top color=white, bottom color=blue!20},
  root/.style     = {treenode, font=\Large, bottom color=red!30},
  env/.style      = {treenode, font=\ttfamily\normalsize},
  dummy/.style    = {circle,draw}
}
\usepackage{wrapfig}
\usepackage{url}
\usepackage{stmaryrd}
\usepackage{physics}
\usepackage{algorithm}
\usepackage{algorithmic}

\usepackage[font={sf}]{caption}
\usepackage[font={sf}]{subcaption}

\newtheorem{thm}{Theorem}[section]

\newtheorem{prop}[thm]{Proposition}

\newtheorem{defn}{Definition}[section]

\newtheorem{rem}{Remark}

\makeatletter
\renewcommand{\boxed}[1]{\text{\fboxsep=.2em\fbox{\m@th$\displaystyle#1$}}}
\makeatother

\title{Towards A Scientific Method for Dynamical Systems}
\author{Vincent Wang\thanks{With thanks to Alan Garfinkel, Pawel Sobocinski, Joel Hancock, Lukas Heidemann, and Gioele Zardini for helpful discussions.}}
\date{%
    Quantum Group, Department of Computer Science, The University of Oxford\\[2ex]%
    \today
}

\begin{document}

\maketitle

\begin{abstract}

The premiss of this paper is the following value proposition:\\

\texttt{Models are good when they describe a system of phenomena, and they are better when they can predict the effect of interventions upon the system.}\\

We introduce a formalism by which dynamical system models can be obtained by specifying basic constituent processes in the form of petri nets, mediated by the Law of Mass Action. We prove the universality of this procedure with respect to ordinary first order differential equations in rational functions. In addition, we introduce a simple graphical procedure for calculating dynamical systems from petri nets in our formalism, which we use to recover several well-known dynamical systems. For proof-of-concept, we use our method to obtain a differential equation model for the HES1 gene transcription network that predicts post-intervention behaviours of the network.

\end{abstract}

\section{Introduction}

When they can be obtained, dynamical systems are extraordinarily expressive descriptive models of complex phenomena \cite{garfinkel_modeling_2017}. However, these models tend to be opaque with respect to what fundamental processes give rise to such systems: often it requires domain expertise and a keen mathematical intuition to glean an understanding of what chemical reactions underlie a system of ordinary first-order differential equations that describe the changing concentrations in a chemical soup over time, and often it cannot be done at all.\\

The epistemic problem is this: anyone, given enough time, can guess at systems of first-order differential equations until one is found that matches up to empirical data `well enough'. How do we know that their model is tracking reality? The practical problem is this: if you have a differential equation model that models a phenomena, and someone asks ``what does the model predict will happen if we do $X$?". How does one answer?\\

Following Popper's characterisation of the scientific method through falsification \cite{thornton_karl_2019}, we at least want to know when a model is wrong. For this, the models in question need not be causal explanations, but must at least carry enough structure to reflect interventions that might be performed in a laboratory. We may then test the model by picking an intervention, performing that intervention in the lab, and then seeing if the intervened mechanism still matches up with empirical data from the intervention condition. Corroborating data cannot tell if the model is 'correct', but if the intervened mechanism and intervened data do not agree, then we can be sure that the model is wrong. If the model stands up to all the interventions we can think of, then it is good insofar as it tracks reality to the extent of all the constituent processes it posits.\\

Guessing a dynamical system that matches some data without an underlying intervention-compatible structure is a nonstarter for the scientific method. If guessing is the only game in town, changing any of the circumstances of the phenomena yields an essentially new phenomenon one has to guess a new model for. This is more stamp-collection than science.\\

On the other hand, combinatorial Petri nets \cite{petri_kommunikation_1962} are excellent prescriptive models of phenomena, that specify phenomena `bottom-up' from constituent processes, in contrast to the descriptive `top-down' flavour of dynamical systems models. If we can systematically translate combinatorial petri nets built from our understanding of constituent processes into dynamical systems models that agree with data obtained from composite and complex processes, we have a foothold to do science. We may attempt to intervene on the consituent processes, both in the abstract petri net and in the laboratory, and obtain new dynamical systems models that purport to describe phenomena under different initial circumstances. We continue a discussion of causality and intervention with respect to dynamical systems in section 2, illustrating the need for a novel approach.\\

In section 3, we show how the Law of Mass Action provides a `natural' conversion of petri nets with rates into dynamical systems models, in the sense that the procedure requires `no added information'; the data of the petri net fully determines the resultant dynamical system. We further develop the mathematical framework to acccommodate a general form of modelling assumption that subsumes the kind of mathematical manoeuvre familiar to practicing mathematical modellers. We conclude the section with a universality result: any system of ordinary first-order differential equations in rational functions can be obtained from a petri net with rates, by applying law of mass action and appropriate modelling assumptions.\\

Non-technical readers unfamiliar with category theory may skim section 3, as in section 4 we introduce a simple graphical procedure to obtain dynamical systems from petri nets that mirrors the computations of our mathematical framework. Using this graphical procedure, we demonstrate how several well-known non-spatial dynamical systems can be recovered from petri net specifications with clear operational interpretations. For engineers, we also include pseudocode for the procedure in the appendix material.\\

In section 5, we apply our graphical procedure to a petri net schematic of the processes underlying the HES1 gene transcription network \cite{hirata_oscillatory_2002}, and we show how, provided empirical data, we may obtain concrete dynamical systems models that are robust to laboratory interventions: our proof-of-concept for a scientific method for dynamical systems.

\section{Causal Models}

Due to Pearl \cite{pearl_causality_2000}, we have a robust repertoire of conceptual tools to perform causal analysis on certain classes of models, including modelling the effect of intervention. A simple example of a Pearlian structural causal model is the diagram $A \rightarrow B \rightarrow C$, where $A,B,C$ are variables of a system. Such schema are realised as models when each arrow is assigned a function, which, taking a value from its input, \emph{uniquely determines} the value of its output.\\

For example, let $A$ be the outcome of a coinflip, $B$ the choice of two paths a fool might take, the left leading to a pub and the right into a ravine, and $C$ whether the fool has a good time in the former or falls to their death in the latter. The coinflip is modelled by an arrow $X \rightarrow A$, where $X$ is a random variable exogenous to the causal system in question. Say that the causal function $A \rightarrow B$ is the fool's pact to themselves that they will take the left path if the coin comes up heads, and the right if the coin comes up tails. And say {}that the causal function $B \rightarrow C$ is the outcome mapping: 

\begin{align*}
\texttt{left}\mapsto\texttt{good time in pub}\\
\texttt{right}\mapsto\texttt{death in ravine}
\end{align*}

The causal model provides a framework for analysis: if we see the fool dead in the ravine, we can meaningfully evaluate the semantics of, for instance, the counterfactual statement \texttt{The fool wouldn't have died if the coin came up heads.} We can also consider interventions in the causal system. The Pearlian notion of intervention finds expression through the ``do" operator, which performs surgery \cite{jacobs_causal_2019} on the arrows of a causal graph; to $\texttt{do}(Y=y)$ with respect to a causal graph is to cut all incoming arrows of the node $Y$, and replace them with a single constant function $y$ going into $Y$. This is made clearer by example: we will play the hand of fate to save the fool.\\

We have three ways to save the fool by a Pearlian intervention. The first is to ensure that the coin comes up heads: $\texttt{do}(A = \texttt{heads})$.  Schematically, this cuts the arrow $X \rightarrow A$, and replaces it by the constant function $X \mapsto \texttt{heads}$. One semantic interpretation is that we have the powers to manipulate chance, and another is that we replace the coin with one that is heads on both sides.\\

The second is to force the mind of the fool: $\texttt{do}(B = \texttt{left})$. Schematically, we cut the $A \rightarrow B$ arrow, and replace it by the constant function $A \mapsto \texttt{left}$. One semantic interpretation is that we distract the fool from their pact and lure them down the left path.\\

Finally, we might force the fool to end up in the pub regardless of the path they take: $\texttt{do}(C = \texttt{good time in pub})$. Schematically, we cut the $B \rightarrow C$ arrow, and replace it by the constant function $C \mapsto \texttt{good time in pub}$. One semantic interpretation is that we set up a detour on the right path leading to the left, and a meaner one is that we knock the fool unconscious whichever path they take, and bring them to the pub to regain consciousness.

\subsection{The unfitness of Pearlian Intervention for dynamical systems}

Pearlian structural causal models and Pearlian intervention are fatally unfit for direct application to dynamical systems. Pearlian causal graphs are necessarily directed acyclic, due to the implicit temporality of causality: at least in classical physics, we disallow events in the future to have bearing on events in the past. This raises a problem when we consider even the most basic dynamical systems. Take a simple logistic growth model, with two variables $P$ (for population) and $R$ (for resources). We might identify two essentially deterministic processes at play:

\begin{enumerate}
\item{Population reproduction rate is determined by the current population size and the amount of available resource. Reproduction produces more population and consumes resource.}
\item{Population death rate is determined by the current population. Death reduces population and generates resource.}
\end{enumerate}

There is no evident way to condense this setup into a structural causal model. Suppose we unfold the graph in time, obtaining an infinite hypergraph:

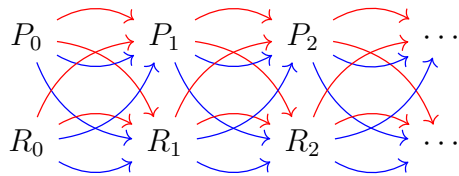
\begin{figure}[h]
\centering
\begin{tikzcd}
P_0 \arrow[r,bend left,red] \arrow[dr, bend left,red] \arrow[r,bend right,blue] \arrow[dr, bend right, blue] & P_1 \arrow[r,bend left,red] \arrow[dr, bend left,red] \arrow[r,bend right,blue] \arrow[dr, bend right, blue] & P_2 \arrow[r,bend left,red] \arrow[dr, bend left,red] \arrow[r,bend right,blue] \arrow[dr, bend right, blue] & \cdots \\
R_0 \arrow[r, bend left, red] \arrow[ur, bend left, red] \arrow[r, bend right, blue] \arrow[ur, bend right, blue]& R_1 \arrow[r, bend left, red] \arrow[ur, bend left, red] \arrow[r, bend right, blue] \arrow[ur, bend right, blue] & R_2 \arrow[r, bend left, red] \arrow[ur, bend left, red] \arrow[r, bend right, blue] \arrow[ur, bend right, blue] & \cdots
\end{tikzcd}
\caption{The reproduction and death processes retain the same colours, and are dashed for resources for the sake of legibility. The numerical subscripts on $P$ and $R$ are discrete time indices}
\label{unfoldproblem}
\end{figure}

In the `unfolded' transcription, temporality is made explicit, but with an essentially arbitrary choice of discrete timestep, which comes with sacrifices, \emph{e.g.} to continuity.\\

The real problem comes with the notion of Pearlian intervention. What would it mean to $\texttt{do}(P = 123)$? The simplest interpretation is that we somehow hold the population constant at 123. In the unfolded interpretation, even if we are clever, we have to cut an infinity of edges, one for each timestep. And whatever formalism we choose, we are asking too much of the experimentalist: we are asking them to summon a Maxwell's Demon in the system that conspires to remove and include members of the population as required at all times to keep the population constant at 123.

\subsection{Our approach}

If we take a more na\"{i}ve approach, a `flat' transcription of the processes into a directed hypergraph (\emph{i.e.} a Petri Net) poses a different challenge: we must retrieve an essentially continuous dynamical system from it.

\begin{figure}[h]
\centering
\begin{tikzcd}
& \texttt{[Reproduction]} \arrow[dl,blue] \arrow[dr, blue] &\\
P \arrow[tail,ur,bend right,blue] \arrow[tail, dr,bend right,red] & & R \arrow[tail, ul,bend left,blue] \arrow[tail, dl, bend left, red]\\
& \texttt{[Death]} \arrow[ur,red] \arrow[ul,red] &
\end{tikzcd}
\caption{The reproduction process is in blue, and the death process is in red. Each process is a directed hyperedge, taking two inputs and returning two outputs.}
\label{flatproblem}
\end{figure}
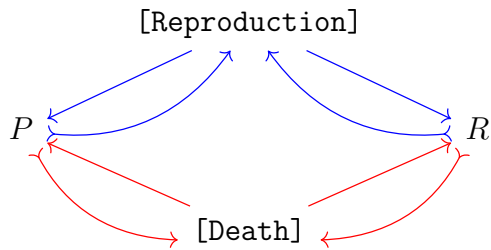

When faced with the real practice of analysing biochemical systems, the only reasonable interventions the experimentalist can perform are to add new basic processes via introduction of new species. If the domain is understood particularly well, then new species can be introduced to selectively inhibit or promote the rates of extant processes, which corresponds to increasing or decreasing the rates of those processes, respectively. Departing from Pearl, this will be our operant conception of what qualifies as an intervention: a targeted manipulation of the \emph{rates associated with processes in a petri net}, rather than of the quantity of species.\\

In the following sections, we will show how dynamical systems on one free variable arise `naturally' from petri nets with rates, in service of a demonstration of the power of this approach on real data obtained from the HES1 gene network.

\section{Law of Mass Action, Categorically}

We will be satisfied with an informal formulation of The Law of Mass Action. Used in mathematical models for chemistry \cite{erdi_mathematical_1989}, it states that the rate of an elementary reaction is proportional to the product of the concentrations of its inputs, each raised to the power of the number of occurrences in the input.\\

We conceive of the Law as a method to obtain algebraically specified dynamical systems on a state space $k^S$ from labelled petri nets with rates on state variables $S$, \emph{without adding any more information}.

\begin{defn}[Labelled Petri Net with Rates]

We work in the small category Set. Given finite sets of \textbf{state variables} $S$ and \textbf{transitions} $T$, a \textbf{base field} $k$ of characteristic $0$, we take a \textbf{a labelled petri net with rate variables} to be a commuting diagram as in Figure \ref{basicpetri}.

\begin{figure}[h]
\centering
\begin{tikzcd}
&k &\\
\mathbb{N}^S &T \arrow[l,"i"] \arrow[u,"r"] \arrow[r,"o"] &\mathbb{N}^S
\end{tikzcd}
\caption{The composition of mass-respecting spans by coproduct. The natural isomorphisms that come wtih the coproduct make this impoverished notion of composition commutative.}
\label{basicpetri}
\end{figure}

\end{defn}

Reflecting common approaches to specifying petri nets \cite{master_generalized_2019}, the object $T$ serves as an object of labelled transitions. $T$ is the peak of a span $\mathbb{N}^S \overset{i}{\leftarrow} T \overset{o}{\rightarrow} \mathbb{N}^S$, which assigns input and output types (with multiplicity) for each labelled transition. The arrow $T \overset{r}{\rightarrow} k$ assigns a \textbf{rate constant} in the base field to each transition.

\begin{defn}[Dynamical System]\footnote{We adopt this definition in part for pragmatic considerations: for almost all cases of interest, one cannot hope to find explicit analytic solutions. Knowing how to compute the change-vector field over the state space is sufficient to begin numerical simulation.}
We take \textbf{state space} to be $k^S$. Where tangent bundles $Tk^S$ are well-defined (as is the case for Euclidean space $\mathbb{R}^S$), define a \textbf{dynamical system} over $k^S$ to be a map $k^S \rightarrow Tk^S$. In Set, we may consider the exponential object $(Tk^S)^{(k^S)}$ as the space of all dynamical systems on $k^S$. Typically, we care when the infinitesimal tangent vectors in $Tk^S$ vary smoothly over $k^S$ (viewed as a manifold), in which case we call the dynamical system \textbf{continuous}.
\end{defn}

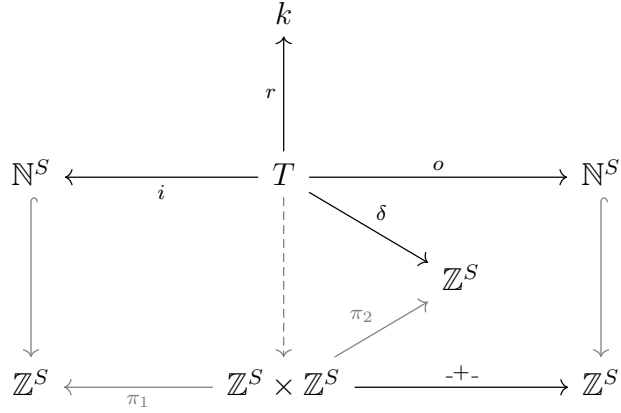
\begin{figure}[h]
\centering
\begin{tikzcd}
& &k & & \\
& & & & \\
\mathbb{N}^S \arrow[dd,hookrightarrow,gray]& &T \arrow[uu,"r"] \arrow[ll,"i"] \arrow[rr,"o"] \arrow[dr,"\delta"] \arrow[dd,gray,dashed]& &\mathbb{N}^S \arrow[dd,hookrightarrow,gray]\\
& & &\mathbb{Z}^S & \\
\mathbb{Z}^S& &\mathbb{Z}^S \times \mathbb{Z}^S \arrow[ll,"\pi_1",gray] \arrow[ur,"\pi_2",gray] \arrow[rr,"\_ + \_"]& & \mathbb{Z}^S
\end{tikzcd}
\caption{We notate the free commutative monoid on $S$ as $\mathbb{N}^S$, and the free commutative group as $\mathbb{Z}^S$. The grey arrows are either universal, or are canonical embeddings. We consider $\mathbb{Z}^S$ to be a group object, with multiplication $\mathbb{Z}^S \times \mathbb{Z}^S \overset{\_ + \_}{\rightarrow} \mathbb{Z}^S$}
\label{dynamicalsystem}
\end{figure}

 The arrow $T \overset{\delta}{\rightarrow} \mathbb{Z}^S$ assigns to each transition the \textbf{net change} in species that results after a single firing of that transition. Collectively, this is precisely the data required for the Law of Mass Action as specified in \cite{baez_compositional_2017}.\\

\begin{defn}[Law of Mass Action]
Given a labelled petri net with rates, the Law of Mass Action is specified by four families of maps\footnote{all canonical or close enough.}.
\begin{enumerate}
\item{The monoid-homomorphism embedding from the free monoid on $S$ to the polynomial ring $k[S]$, $\mathbb{N}^S \hookrightarrow k[S]$, which is (punning) identity on objects, and maps monoid multiplication on $\mathbb{N}^S$ to polynomial multiplication in $k[S]$. Used in Figure \ref{growinglegs}.}
\item{The ring-homomorphism embedding $k \hookrightarrow k[S]$. Used in Figure \ref{growinglegs}.}
\item{The `forgetful' embedding from the hom-object $k[S]^T$ into the free monoid $\mathbb{N}^{k[S]}$, obtained by counting arrows into each $k[S]$ in the graph of $f: T \rightarrow k[S]$. Used in Figure \ref{TSstage2}.}
\item{The embedding $k[S]^S \hookrightarrow (Tk^S)^{(k^S)}$ which maps $S$-indexed assignments $f: S \rightarrow k[S]$ to maps $g: k^S \rightarrow Tk^S$, by specifying for each $\sigma_i \in S$ that $\dot{\sigma_i} := f_i \in k[S]$. We choose a basis ordering for $Tk^S$, viewed as a vector space, so $(\dot{\sigma_1},\dot{\sigma_2},\cdots,\dot{\sigma}_{|S|}) := (f_1,f_2,\cdots,f_{|S|})$ defines the map $k^S \rightarrow Tk^S$. Used in Figure \ref{TSstage3}.}
\end{enumerate}

The presence of these maps allows us to recover a unique continuous dynamical system on $k^S$ from a labelled petri net with rates, following Figures \ref{growinglegs}, \ref{spanlegsmerge}, \ref{TSstage}, \ref{TSstage2}, and \ref{TSstage3}.

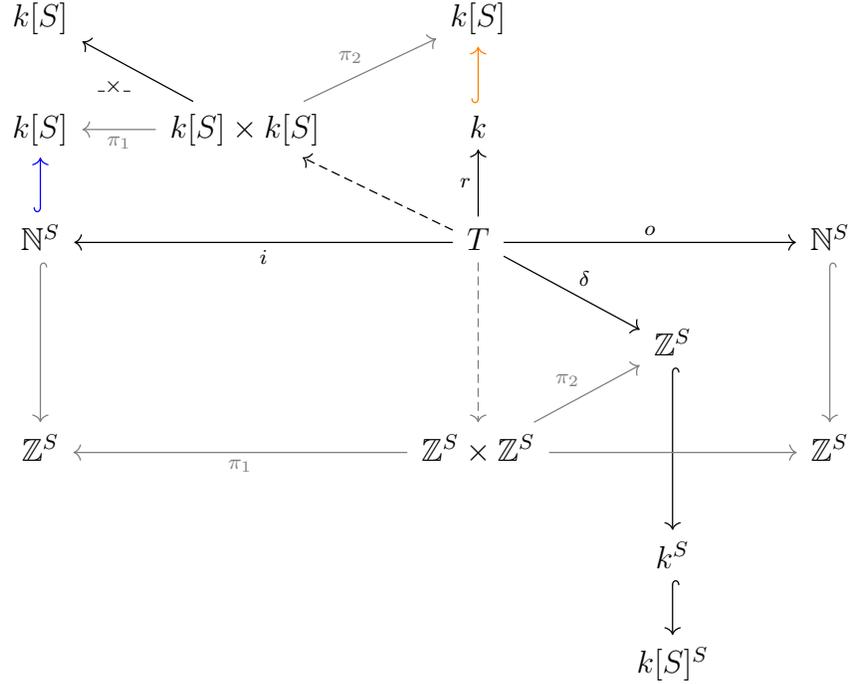
\begin{figure}[h]
\centering
\begin{tikzcd}
k[S]& &k[S] & & \\
k[S] &k[S] \times k[S] \arrow[l,"\pi_1",gray]  \arrow[ur,"\pi_2",gray] \arrow[ul,"\_ \times \_"]&k \arrow[u,hookrightarrow,orange] & & \\
\mathbb{N}^S \arrow[u,hookrightarrow,blue] \arrow[dd,hookrightarrow,gray]& &T \arrow[ul,dashed] \arrow[u,"r"] \arrow[ll,"i"] \arrow[rr,"o"] \arrow[dr,"\delta"] \arrow[dd,gray,dashed]& &\mathbb{N}^S \arrow[dd,hookrightarrow,gray]\\
& & &\mathbb{Z}^S \arrow[dd,hookrightarrow] & \\
\mathbb{Z}^S& &\mathbb{Z}^S \times \mathbb{Z}^S \arrow[ll,"\pi_1",gray] \arrow[ur,"\pi_2",gray] \arrow[rr,gray]& & \mathbb{Z}^S\\
& & &k^S \arrow[d,hookrightarrow] & \\
& & &k[S]^S &
\end{tikzcd}
\caption{Map 1 (in blue) and map 2 (in orange), in conjunction with the polynomial muliplication map $k[S] \times k[S] \overset{\_ \times \_}{\rightarrow} k[S]$ and a sequence of embeddings $\mathbb{Z}^S \hookrightarrow k^S \hookrightarrow k[S]^S$ lifted from the sequence of canonical ring-homomorphism embeddings $\mathbb{Z} \hookrightarrow k \hookrightarrow k[S]$ allow us to obtain a span $k[S] \leftarrow T \rightarrow k[S]^S$.}
\label{growinglegs}
\end{figure}

\begin{figure}[h]
\centering
\begin{tikzcd}
& T \arrow[dl,"c"] \arrow[dr,"d"] & \\
k[S]& &k[S]^S
\end{tikzcd}
\caption{We label the legs of the span $c$, for `coefficient', and $d$, for `difference'. By construction each transition $\tau$ is assigned a polynomial rate coefficient equal to the product of that transition's rate constant $r(\tau)$ and a polynomial product of species $\sigma \in S$ each raised to the power of its number of occurrences in the input of $\tau$. We call this the \textbf{mass-respecting rate}. Further, each transition is assigned a change vector $o(\tau) - i(\tau) \in \mathbb{Z}^S \hookrightarrow k[S]^S$, whose $j^{th}$ component is the net change in the $j^{th}$ species $\sigma_j \in S$ resulting from a single firing of transition $\tau$.}
\label{spanlegsmerge}
\end{figure}
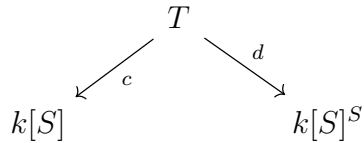

\begin{figure}[h]
\centering
\begin{tikzcd}
& T \times S \arrow[d,dashed] \arrow[dl,"c \circ \pi_1"] \arrow[dr,"ev \circ \langle d . 1_S \rangle"] & \\
k[S]& k[S] \times k[S] \arrow[l,gray,"\pi_1"] \arrow[r,gray,"\pi_2"] \arrow[d,"\_ \times \_"]&k[S]\\
& k[S] &
\end{tikzcd}
\caption{Taking advantage of the presence of exponentials in Set, we obtain the above commuting map by uncurrying $k[S]^S$. Postcomposing by the polynomial multiplication map yields a morphism $T \times S \rightarrow k[S]$, which maps a (transition, state variable) pair $(\tau,\sigma)$ to the net change in $\sigma$ resulting from $\tau$, multiplied by the appropriate mass-respecting rate.}
\label{TSstage}
\end{figure}

\begin{figure}[h]
\centering
\begin{tikzcd}
S \arrow[r] & k[S]^T \arrow[r,blue,hookrightarrow] & \mathbb{N}^{k[S]} \arrow[r,"+"] & k[S]
\end{tikzcd}
\caption{Currying $T$ and postcomposing with map 3 (in blue) allows us to reach the evaluation morphism for the polynomial addition algebra. The resulting morphism $S \rightarrow k[S]$ forgets the contribution of individual transitions, having summed them together. Each species $\sigma$ is assigned its sum change due to all transitions in $T$.}
\label{TSstage2}
\end{figure}

\begin{figure}[h]
\centering
\begin{tikzcd}
\{\star\} \arrow[r] & k[S]^S \arrow[r,blue,hookrightarrow] & {(Tk^S)}^{(k^S)}
\end{tikzcd}
\caption{Finally, currying $S$ and applying map 4 (in blue) yields the desired element. We may uncurry $k^S$ to obtain the dynamical system $k^S \rightarrow Tk^S$.}
\label{TSstage3}
\end{figure}

\end{defn}

\clearpage

\begin{rem}

Though our interest is not in investigating the compositionality of these structures, we mention here that our labelled petri nets with rates may compose cheaply by taking coproducts over the spans obtained in Figure \ref{TSstage} as in Figure \ref{composingspanlegs}. It is easy to verify that the required diagrams as in Figure \ref{dynamicalsystem} commute, since all objects in the diagram are reachable from $T$, such that the resulting petri net with rates is $T+T' \overset{i + i'}{\rightarrow} \mathbb{N}^S$, $T+T' \overset{o + o'}{\rightarrow} \mathbb{N}^S$, $T+T' \overset{r + r'}{\rightarrow} \mathbb{Z}^S$.

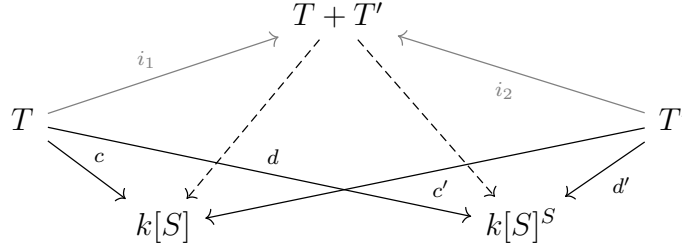
\begin{figure}[h]
\centering
\begin{tikzcd}
& & T + T' \arrow[ddr,dashed] \arrow[ddl,dashed]& &\\
T \arrow[urr,"i_1",gray] \arrow[dr,"c"] \arrow[drrr,"d"]& & & &T' \arrow[dl,"d'"] \arrow[dlll,"c'"] \arrow[ull,"i_2",gray]\\
&k[S] & &k[S]^S &
\end{tikzcd}
\caption{The composition of mass-respecting spans by coproduct. The natural isomorphisms that come wtih the coproduct make this impoverished notion of composition commutative.}
\label{composingspanlegs}
\end{figure}

\end{rem}

\subsection{A (limited) Universality Result}

Assume an arbitrary but fixed finite set of state variables $S$. For any $S$-indexed collection of polynomials in $\mathbb{R}[S]$, we obtain a system of ordinary first-order differential equations with polynomial expressions in $\mathbb{R}[S]$ for each $\dot{\sigma}$ ($\sigma \in S$), and thus dynamical system $\mathcal{D}$ in one variable over state space $\mathbb{R}[S]$. Let $\mathcal{D}(\sigma)$ denote the polynomial in $\mathbb{R}[S]$ that defines $\dot{\sigma}$ for the dynamical system $\mathcal{D}$. Let $\mathcal{P}$ denote a labelled petri net with rates, and $\mathcal{M}\mathcal{P}$ denote the dynamical system obtained from $\mathcal{P}$ by the mass-action procedure defined previously.

\begin{prop}
If for all $\sigma \in S$, $\mathcal{D}(\sigma)$ is a polynomial in $\mathbb{R}[S]$ such that every additive term with a negative coefficent contains a positive integer power of $\sigma$, then there exists a (not necessarily unique) $\mathcal{P}$ such that $\mathcal{M}\mathcal{P} \equiv \mathcal{D}$.
\label{thm1}
\end{prop}

\begin{proof}
We provide a constructive proof. We pick some $\sigma \in S$ such that $\mathcal{D}(\sigma) \neq 0$ (or else we are done, by setting $T = \{\star\}$, $r: \{\star\} \mapsto 0$, and $i,o: \{\star\} \mapsto 0$).\\

Without loss of generality, we write: $$\mathcal{D}(\sigma) \equiv \sum\limits_{i} (r_i \prod\limits_{\sigma_j\in J_i \subseteq S}\sigma_j^{n(\sigma_j)})$$ where $r_i \in \mathbb{R}$, and $n(\sigma_j)$ is the positive integer power of the $\sigma_j$ occurrence. For each term $r_i \prod\limits_{\sigma_j\in J_i \subseteq S}\sigma_j^{n(\sigma_j)}$, create a transition label $\tau_i$, and expand $r,i,o$ by the following data:

\begin{align*}
&r(\tau_i) := r_i \\
&i(\tau_i) := \begin{cases} \sigma_k \mapsto n(\sigma_k) & (k \in J_i) \\ \sigma_k \mapsto 0 & (\text{otherwise}) \end{cases} \ : \mathbb{N}^S\\
&o(\tau_i) := \begin{cases} \sigma' \mapsto i(\tau_i)(\sigma') & (\sigma' \neq \sigma) \\
 \sigma \mapsto i(\tau_i)(\sigma) + 1 & (r_i > 0) \\  
 \sigma \mapsto i(\tau_i)(\sigma) - 1 & (r_i < 0)
 \end{cases} \ : \mathbb{N}^S
\end{align*}

By construction, each $\tau_i$s contribution to the net change of every state variable apart from the chosen $\sigma$ is 0, and it is easy to calculate that the outcome of the mass-action procedure for a petri net consisting of $\tau_i$ (associated data) alone yields $\mathcal{M}\mathcal{P}(\sigma) = \mathcal{D}(\sigma)$.
\end{proof}

\subsection{Modelling modelling assumptions}

Oftentimes in practice, one wishes to employ modelling assumptions, most typically in a form that expresses certain state variables in terms of others. These assumptions simplify the model by reducing the dimension of the state space.\\

For instance, where the system in question is assumed to conserve some quantity such as the sum concentration of two species $A + B$, we might wish to define $B := l - A$ for some limiting constant $l$.\\

It may be that rates for processes are dependent on `derivative' quantities. For instance, rates for some basic processes in the Hodgkin-Huxley model \cite{hodgkin_quantitative_1952} of action-potential propagation in neurons are proportional to the voltage available in the system, which is best expressed as the \emph{difference} in the available sodium and potassium ions.

\begin{defn}[Modelling Assumption]
A \textbf{modelling assumption} for a labelled petri net with rates with state variables $S$ consists of:
\begin{itemize}
\item{A choice of subset $S' \subset S$ of variables targeted for removal}
\item{An assignment of interpretations $S' \rightarrow k[(S-S')]$}
\end{itemize}
\end{defn}

Tacitly, when a modelling assumption of this kind is made, the practitioner is \emph{throwing away} the input and output edges for each of the state variables they have removed: by defining voltage to be difference in ions $V := N - K$, the modeller ceases to consider the contributing effect to $\dot{V}$ from the system, as $V$ and $\dot{V}$ are already fully determined. Removing transitions in this manner reminiscent of Pearlian intervention in structural causal models \cite{jacobs_causal_2019}.\\

\begin{defn}[Labelled Petri Net with Rates and Modelling Assumptions]
Given a labelled petri net with rates and modelling assumptions, we may obtain a new labelled petri net with rates as in Figure \ref{newpetri}.\\

The modellers' assignment choice $S' \rightarrow k[S-S']$ induces a polynomial ring morphism $k[S] \rightarrow k[S-S']$ by extending $S' \rightarrow k[S-S']$ with identity maps on $S-S'$ to obtain $S \rightarrow k[S-S']$, which lifts to $k[S] \rightarrow k[k[S-S']] = k[S-S']$. These are the orange maps in the diagram.\\

The blue maps are obtained by mapping occurrences of $\sigma \in S'$ to the additive identity, dependent on the modeller's choice $S'$.

\begin{figure}[h!]
\centering
\begin{tikzcd}
k[S-S']& &k[S-S'] & & \\
& k[S-S'] \times k[S-S'] \arrow[ul,"\_ \times \_"] \arrow[dl,gray,"\pi_1"] \arrow[ur,gray,"\pi_2"]&k[S]  \arrow[u,->>,orange] & & \\
k[S-S']&k[S] \arrow[l,->>,orange] &k \arrow[u,hookrightarrow] & & \\
\mathbb{N}^{(S-S')} \arrow[dd,hookrightarrow,gray]& \mathbb{N}^{S} \arrow[u,hookrightarrow] \arrow[l,->>,blue] &T \arrow[uul,dashed] \arrow[u,"r"] \arrow[l,"i"] \arrow[r,"o"] \arrow[dr,"\delta"] \arrow[dd,gray,dashed]& \mathbb{N}^S \arrow[r,->>,blue] &\mathbb{N}^{(S-S')} \arrow[dd,hookrightarrow,gray]\\
& & &\mathbb{Z}^{(S-S')} \arrow[dd,hookrightarrow] & \\
\mathbb{Z}^{(S-S')}& &\mathbb{Z}^{(S-S')} \times \mathbb{Z}^{(S-S')} \arrow[ll,"\pi_1",gray] \arrow[ur,"\pi_2",gray] \arrow[rr,gray]& & \mathbb{Z}^{(S-S')}\\
& & &k^{(S-S')} \arrow[d,hookrightarrow] & \\
& & &k[S]^{(S-S')} &
\end{tikzcd}
\caption{The blue maps effectively cut all input and output edges to all species in $S'$ in the final net, in a manner that does not interfere with the computation of the upper half of the diagram, which assigns mass-respecting rates according to the original edges.}
\label{newpetri}
\end{figure}
\end{defn}

\clearpage

\subsection{A Universality Result}

Recall that the commutative ring of rational functions in variables $S$ over a commutative ring $k$ -- which we shall notate with the boldface $\mathbf{k}[S]$ -- consists of expressions of the form $\frac{p}{q}$, where $p, q$ are polynomial expressions in $k[S]$, and $q \neq 0$. Let $\mathcal{D}$ again denote a dynamical system over state space $k^S$, let $\mathcal{P}_Y$ denote a petri net with rates over variables $X$, and let $\mathcal{A}X$ denote a collection of modelling assumptions with rational function expressions: \emph{i.e.} given some $X \subset Y$, an assignment $(Y-X) \rightarrow \mathbf{k}[X]$. Let $\mathcal{M}\mathcal{P}_{Y}^{\mathcal{A}X}$ denote the dynamical system obtained by the mass-action procedure on some net $\mathcal{P}_Y$ followed by modelling assumptions $\mathcal{A}X$.

\begin{thm}[mass-action and assumptions are universal for rational function dynamical systems]
For any dynamical system $\mathcal{D}$ over $k^X$, such that for all $\sigma \in X$, $\mathcal{D}(\sigma) \in \mathbf{k}[X]$, there exists a (not necessarily unique) petri net with rates $\mathcal{P}_Y$ over a superset of variables $Y \supseteq X$ along with assumptions $\mathcal{A}X$ such that $\mathcal{D} \equiv \mathcal{M}\mathcal{P}_{Y}^{\mathcal{A}X}$.
\label{thm2}
\end{thm}

\begin{proof}
We proceed similarly to Proposition \ref{thm1}, this time creating novel state variables and assumptions as required to handle expressions in the denominator.\\

We pick some $\sigma \in S$ such that $\mathcal{D}(\sigma) \neq 0$ (or else we are done, by setting $T = \{\star\}$, $r: \{\star\} \mapsto 0$, and $i,o: \{\star\} \mapsto 0$).\\

Without loss of generality, we write: $$\mathcal{D}(\sigma) \equiv \sum\limits_{i} (\frac{r_i \prod\limits_{\sigma_j\in J_i \subseteq S}\sigma_j^{n(\sigma_j)}}{\rho_i})$$ where $r_i \in \mathbb{R}$, $n(\sigma_j)$ is the positive integer power of the $\sigma_j$ occurrence, and $\rho_i \in k[X]$. We ask for a rewrite in the special missing case of Proposition \ref{thm1}: terms $\mathcal{D}(\sigma)$ with negative coefficient that do not contain any occurrences of $\sigma$ are to be multiplied by $\frac{\sigma}{\sigma}$.\\

After re-expression, for each $\rho_i \in \mathbf{k}[X]$, create a novel state variable $\gamma_i$, and extend the assumptions with the assignment $\gamma_i \mapsto \frac{1}{\rho_i}$. Doing this for all $\sigma \in X$, we may construct a new dynamical system $\mathcal{D}'$ over state variables $Y \supset X$ extended by the $\gamma$s, such that:

$$\mathcal{D}'(\sigma) \equiv \sum\limits_{i} r_i \prod\limits_{\sigma_j\in J_i \subseteq S}\sigma_j^{n(\sigma_j)}\gamma_i$$

For each $\gamma$, set $\mathcal{D}'(\gamma) := 0$. Applying Proposition \ref{thm1} now yields a petri net $\mathcal{P}_Y$ such that $\mathcal{M}\mathcal{P}_Y \equiv \mathcal{D}'$, and $\mathcal{M}\mathcal{P}_Y^{\mathcal{A}X} \equiv \mathcal{D}$.
\end{proof}

\clearpage

\section{Elementary Dynamical Systems, and a tentative Graphical Calculus}

In this section, we develop a graphical calculus that implements the conversion a labelled petri net with rates into a continuous dynamical system, along with modelling assumptions, which agrees with our specification in the previous section. We proceed by example, recovering some well known dynamical systems on one variable, including a novel derivation of the Allee Effect. Throughout, we make light use of an informal notation reminiscent of linear logic to specify the inputs and outputs of transitions in text.

\subsection{Exponential Growth}

Suppose we have an asexually reproducing thing $T$. We might capture this setup by a single atomic process $T \multimap T \otimes T$, glossed as \texttt{"A thing turning into two things."} Our only state variable is $T$, with a net change $(T + T) - T = T$ at rate proportional to $T$, say by a constant positive coefficient $r$. We can depict this system as in Figure \ref{fig:exp}.

\begin{figure}[h]
\centering
\begin{minipage}{.3\textwidth}
  \centering
  \includegraphics[width=0.8\linewidth]{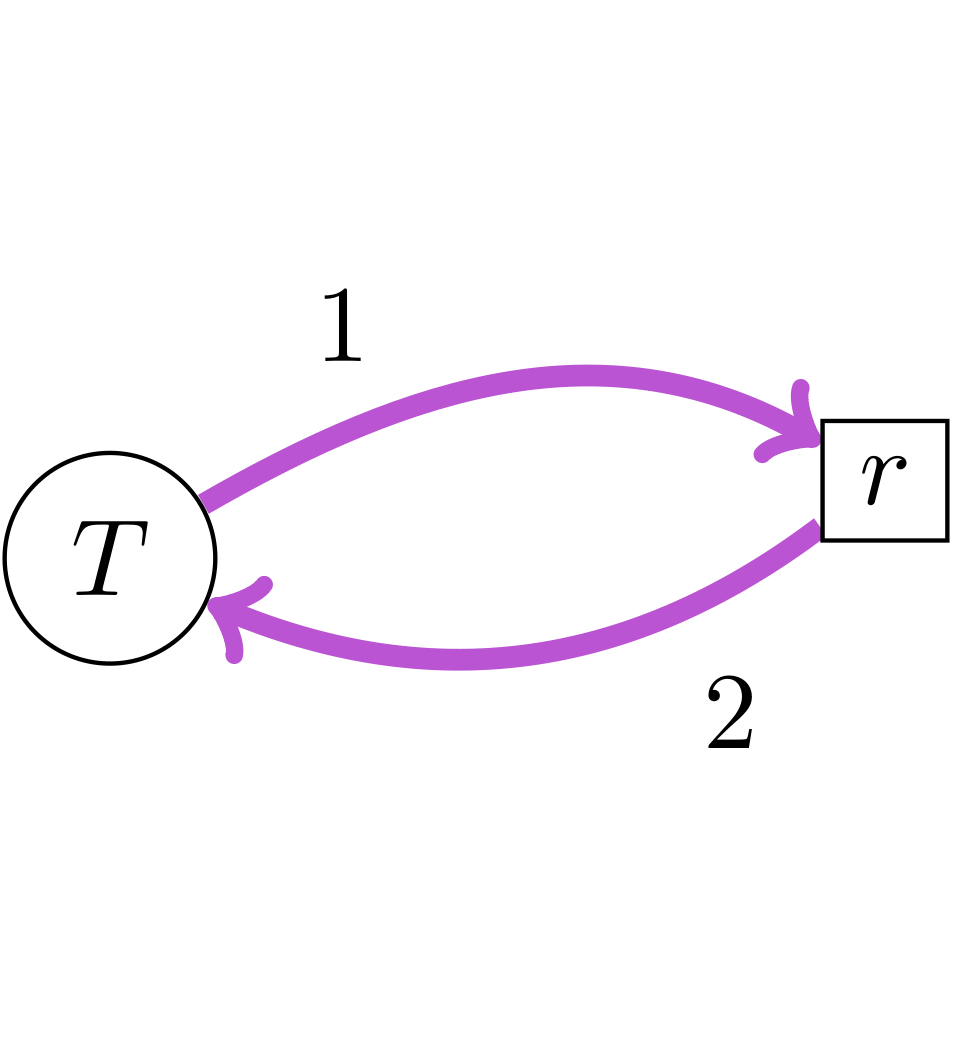}
  \captionof{figure}{}
  \label{fig:exp}{}{}
\end{minipage}
\begin{minipage}{.3\textwidth}
  \centering
  \includegraphics[width=.8\linewidth]{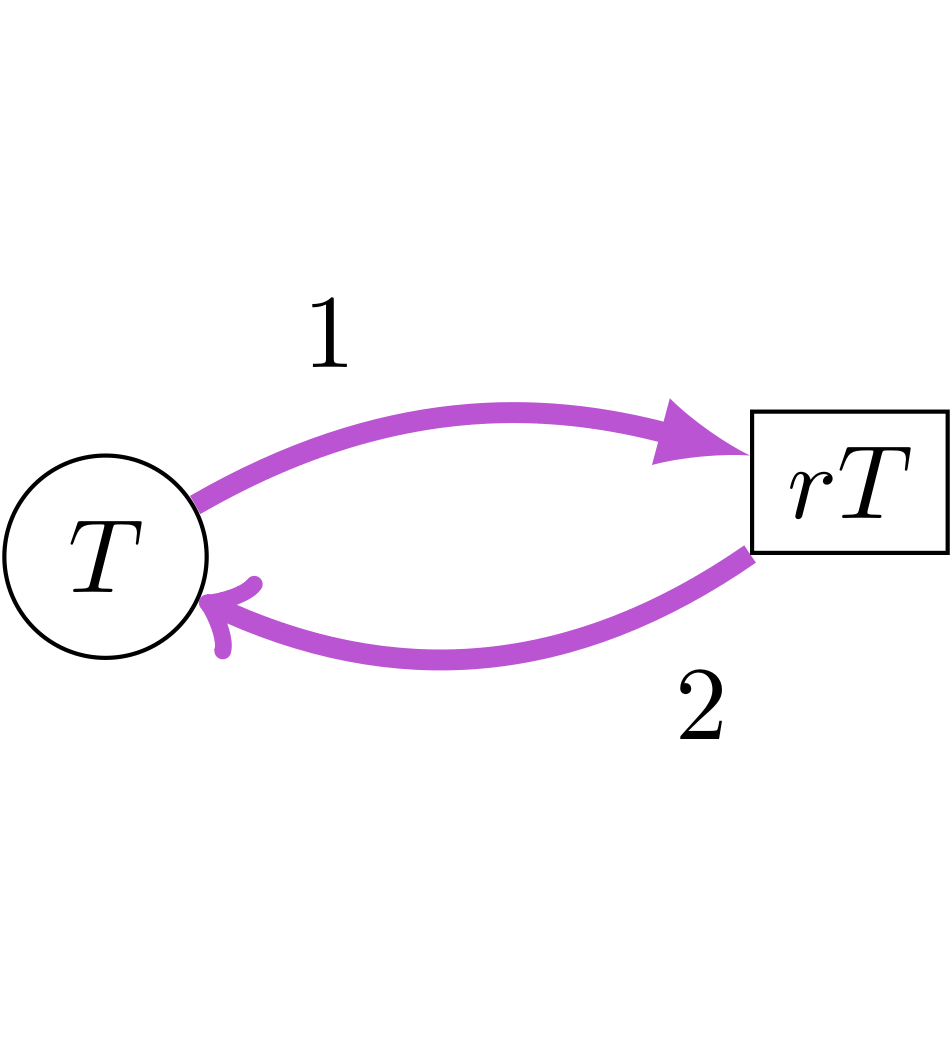}
  \captionof{figure}{}
  \label{fig:exp2}{}{}
\end{minipage}
\begin{minipage}{.3\textwidth}
  \centering
  \includegraphics[width=.8\linewidth]{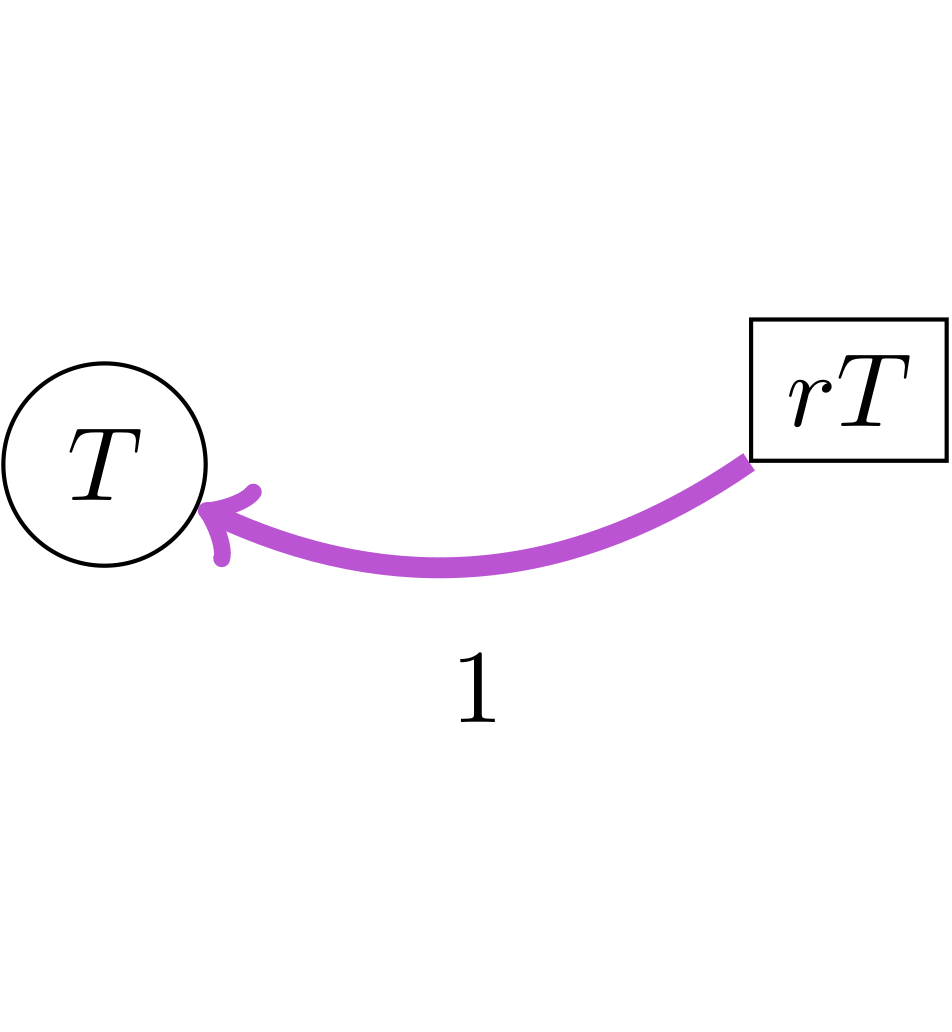}
  \captionof{figure}{}
  \label{fig:exp3}{}{}
\end{minipage}
\caption{We label transitions with boxes decorated by their rate variables, 
  and we label arrows with positive integer multiplicity, as in Figure \ref{fig:exp}. We compute mass-respecting rates for 
  transitions by pushing through species with power equal to mulitplicity of input, 
  marking the arrows we have done, as in Figure \ref{fig:exp2}. Finally, we cancel pairs 
  of arrows of the same colour but going in opposite directions to obtain net change, as in Figure \ref{fig:exp3}.}
\end{figure}

We can read off the dynamical system from the final diagram in Figure \ref{fig:exp3} by setting $\dot{T}$ to be the sum of boxes with edges ingoing into $T$.

$$\dot{T} = rT$$

The analytical solution of the system is $e^{rt}$ ($t$ for time) plus a boundary condition constant $c$. So we have exponential growth, and this is unsurprising, since we have asked for the growth of $T$ to be proportional to itself.\\

Notably, if we vary the starting assumption to idealised sexual reproduction, where the atomic process requires two inputs to produce some number of outputs: \emph{i.e.} $T \otimes T \multimap T \otimes \ldots \otimes T$, the qualititative behaviour of the resulting system remains the same: unstable equilibrium at $T = 0$, and positive growth for $T > 0$.

\clearpage

\subsection{Logistic Growth}

Suppose we amend the previous example with some additional modelling assumptions. If we take two atomic processes to model our system, namely reproduction: \texttt{"Two things meet and make a new thing"} $$T \otimes T \multimap T \otimes T \otimes T$$

and death: \texttt{"A thing disappears"}

$$T \multimap 0$$

Assigning positive coefficients $r$ for reproduction and $d$ for death, we obtain the diagram in Figure \ref{fig:log}. We notate each transition with a different colour.

\begin{figure}[h]
\centering
\begin{minipage}{.3\textwidth}
  \centering
  \includegraphics[width=0.8\linewidth]{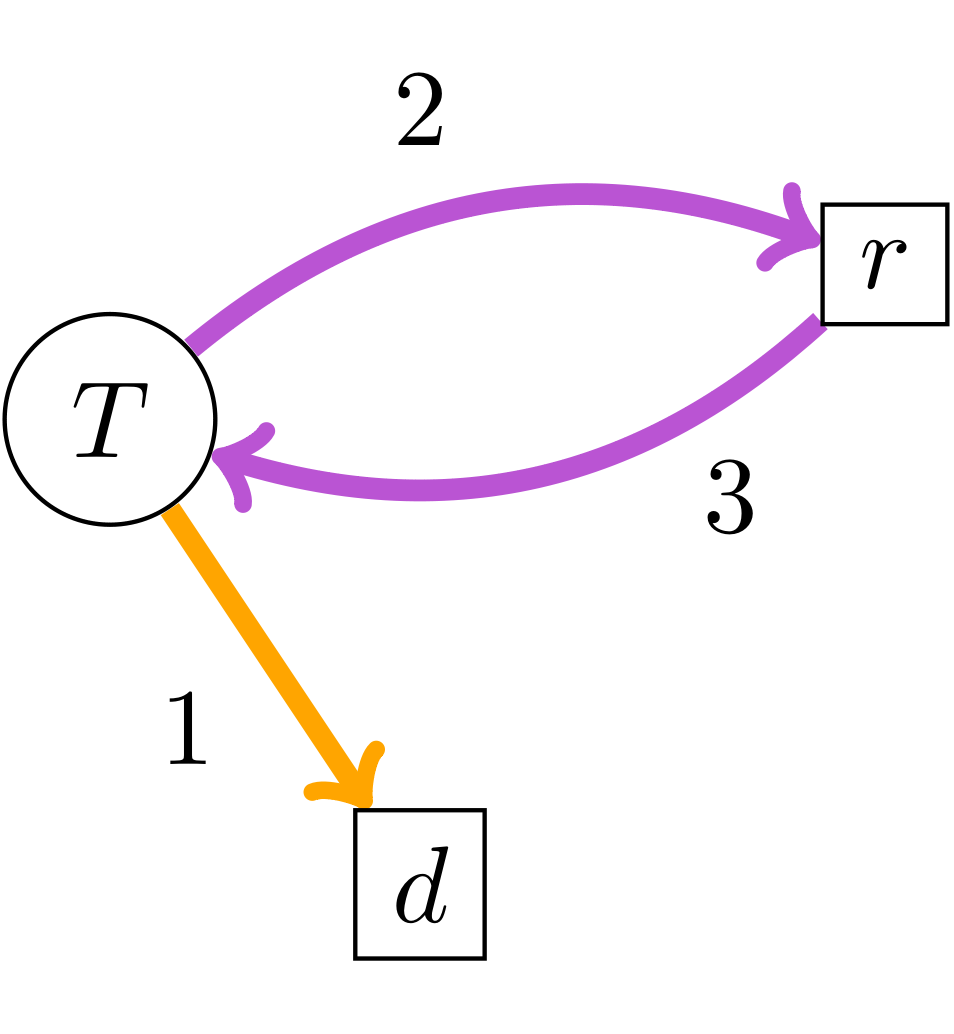}
  \captionof{figure}{}
  \label{fig:log}{}{}
\end{minipage}
\begin{minipage}{.3\textwidth}
  \centering
  \includegraphics[width=.8\linewidth]{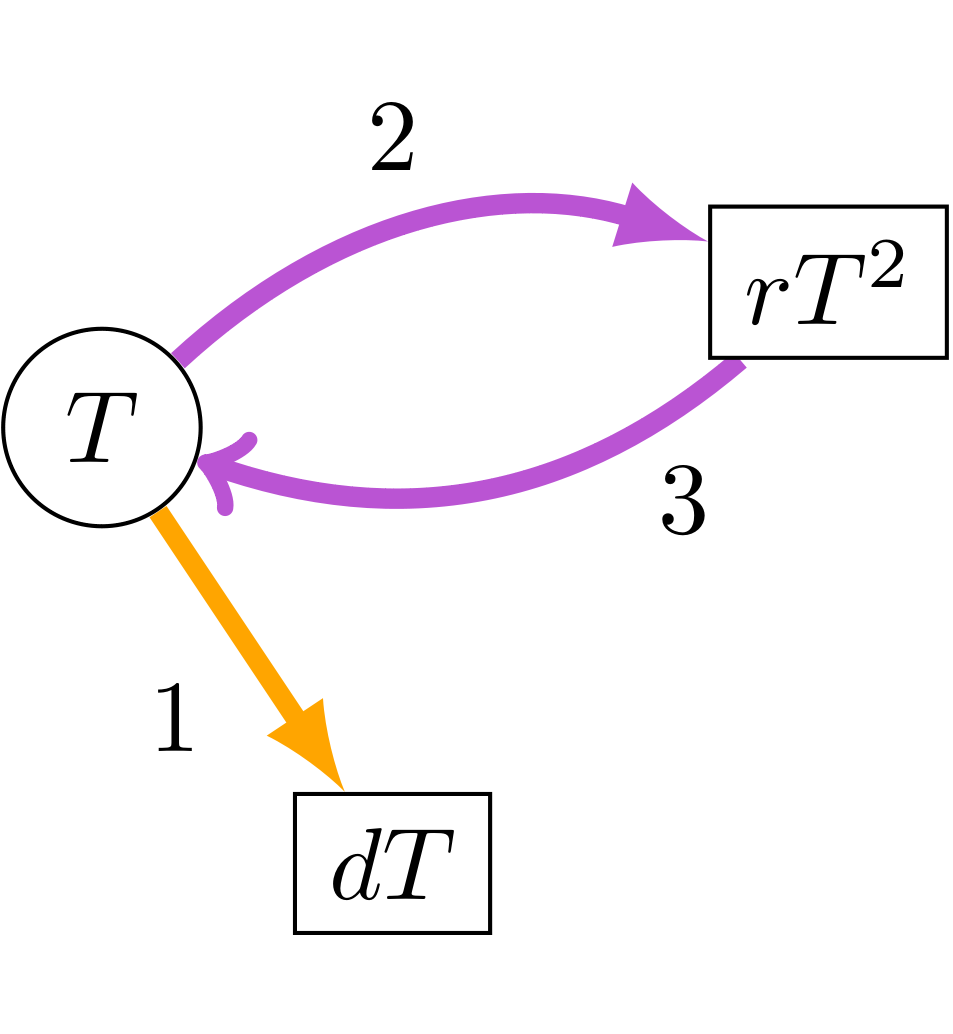}
  \captionof{figure}{}
  \label{fig:log2}{}{}
\end{minipage}
\begin{minipage}{.3\textwidth}
  \centering
  \includegraphics[width=.8\linewidth]{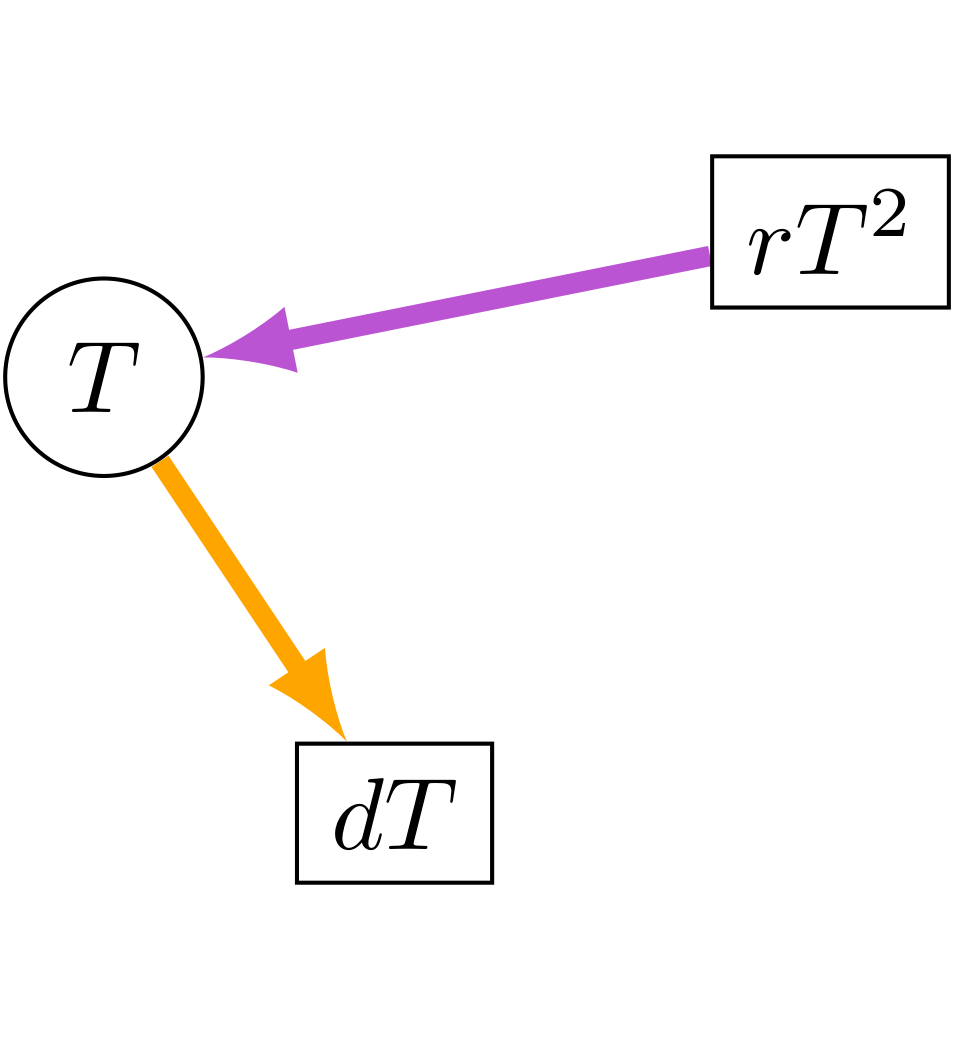}
  \captionof{figure}{}
  \label{fig:log3}{}{}
\end{minipage}
\caption{We obtain the desired result by following the same procedure. 
First pushing through species to obtain mass-respecting rates in Figure \ref{fig:log2},
then cancelling pairs of arrows to obtain Figure \ref{fig:log3}.}
\end{figure}

We can read off the expression for $\dot{T}$ by adding the boxes with incoming, 
and subtracting those with outgoing edges. We obtain from reproduction a positive contribution of magnitude $rT^2$ to $\dot{T}$, and from death a negative contribution of magnitude $dT$ to $\dot{T}$. So the resulting dynamical system is expressed:

$$\dot{T} = rT^2 - dT = rT(\frac{d}{r}-T)$$

So we obtain logistic growth. Notably, we do not obtain logistic growth in this way if we revert the sexual reproduction process to asexual reproduction, which is commonsense: if reproduction and death are both solitary events for each member of the species, the stronger of the two forces results in plain exponential growth (or decay).

\subsubsection{Logistic Growth From Finite Energy}

Now assume asexual reproduction, but amend the process such that reproduction requires energy $E$ available in the thing's environment. We might capture this setup by a single atomic process, which we would gloss as \texttt{"a thing comes by some energy, it consumes the energy and produces a new thing."}. The atomic process that models this gloss is one that consumes one energy and one thing, and returns two things: $E \otimes T \multimap T \otimes T$.\\

So, our state variables are $\{E,T\}$. We assign a weight coefficient $r$ to the reproduction process, obtaining the diagram in Figure \ref{fig:logen}. Now say we also wish to assume that energy is scarce -- finite, and non-replenishing. Perhaps there is only some amount $c$ energy available in the system. This is a consequence of the stronger, but natural, assumption of the system being closed: $\dot{T} + \dot{E} = 0 \implies \dot{E} = - \dot{T} \implies E = c - T$. So we adopt the modelling assumption $E \mapsto (c-T)$, which we demonstrate how to apply in the transition between Figures \ref{fig:logen3} and \ref{fig:logen4}.

\begin{figure}[h]
\centering
\begin{minipage}{.22\textwidth}
  \centering
  \includegraphics[width=0.8\linewidth]{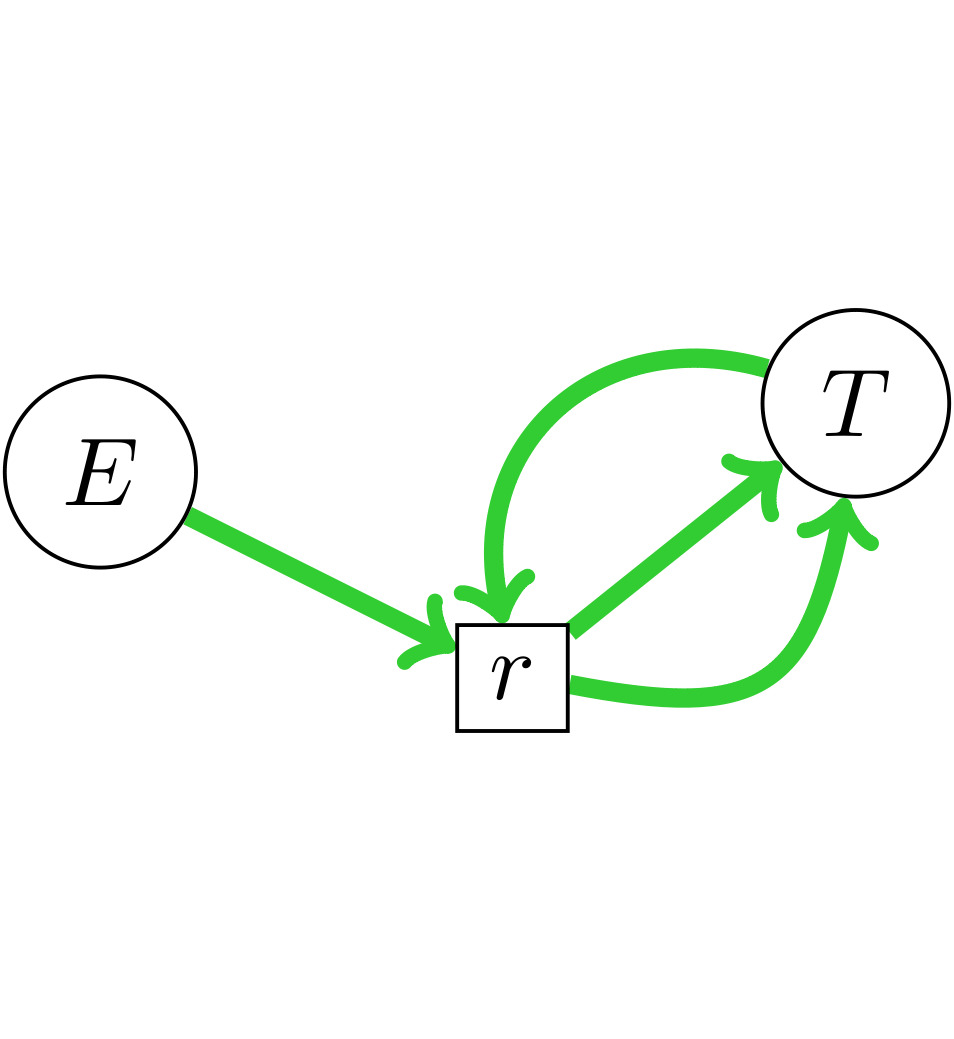}
  \captionof{figure}{}
  \label{fig:logen}{}{}
\end{minipage}
\begin{minipage}{.22\textwidth}
  \centering
  \includegraphics[width=0.8\linewidth]{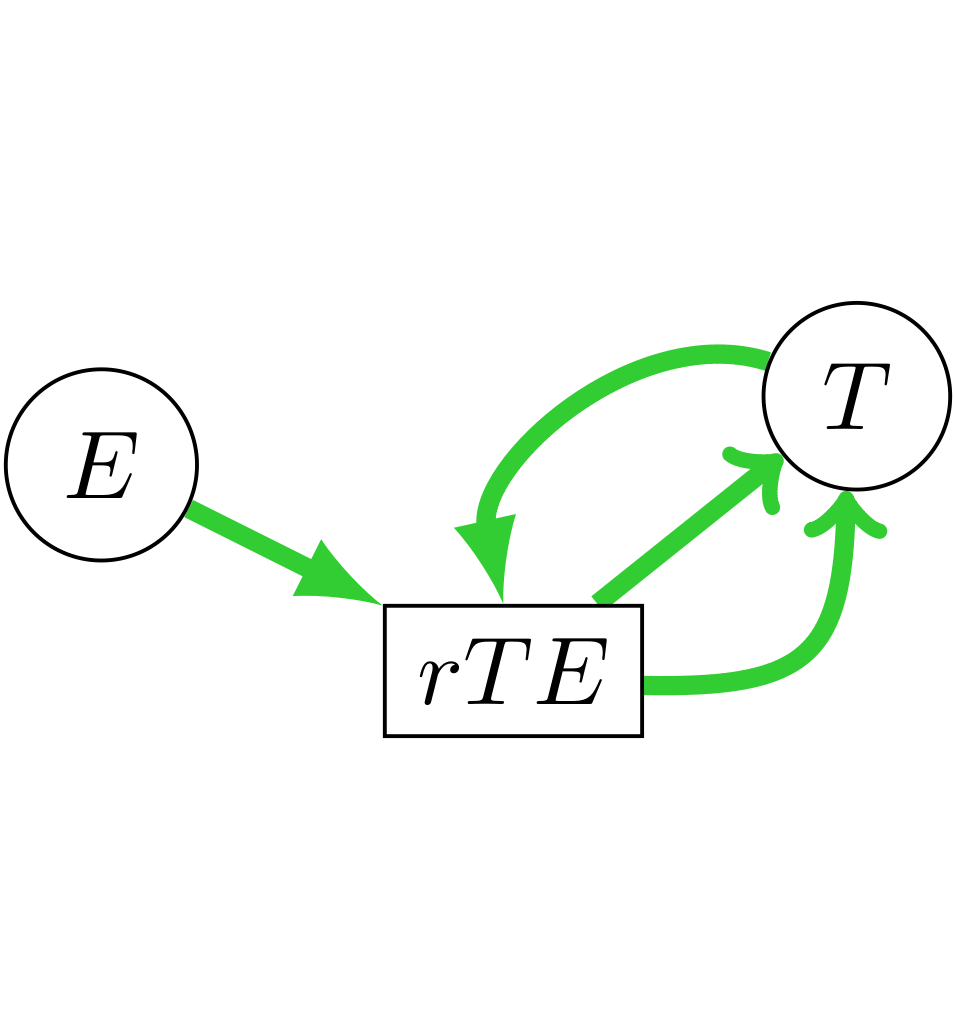}
  \captionof{figure}{}
  \label{fig:logen2}{}{}
\end{minipage}
\begin{minipage}{.22\textwidth}
  \centering
  \includegraphics[width=0.8\linewidth]{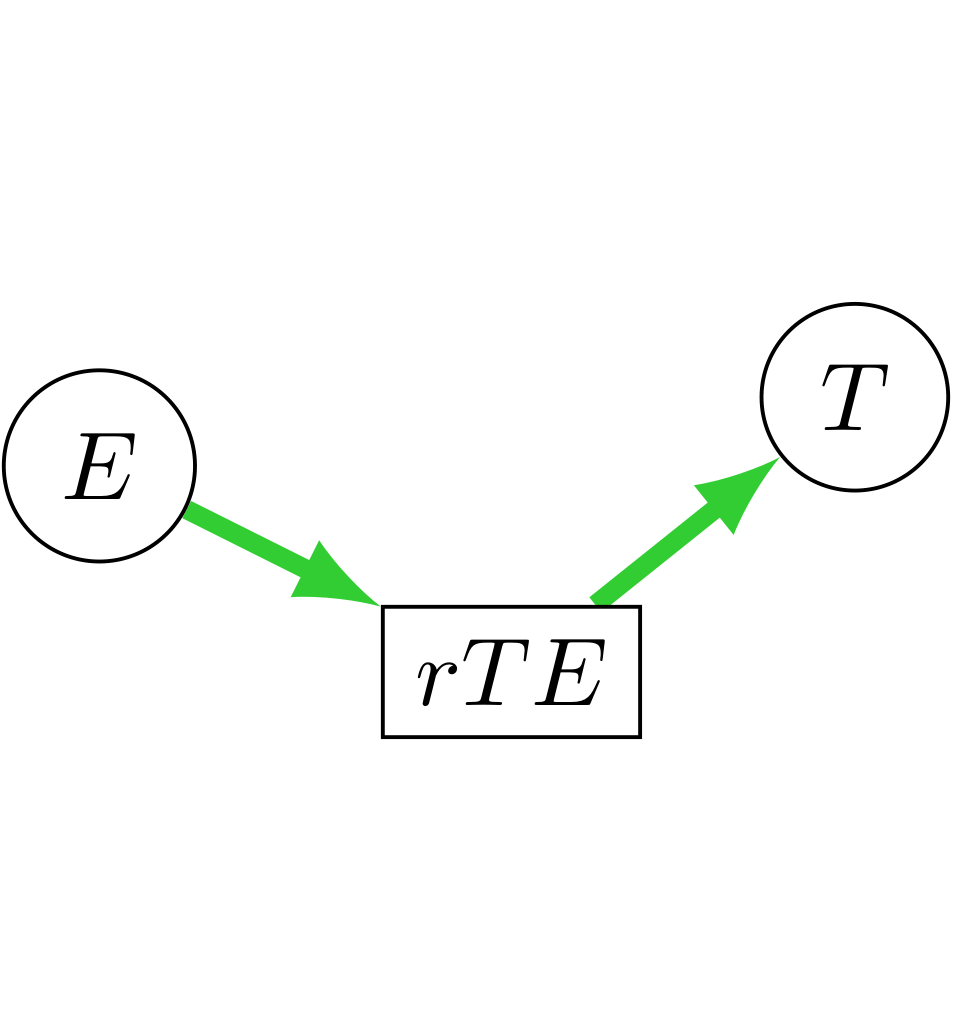}
  \captionof{figure}{}
  \label{fig:logen3}{}{}
\end{minipage}
\begin{minipage}{.22\textwidth}
  \centering
  \includegraphics[width=0.8\linewidth]{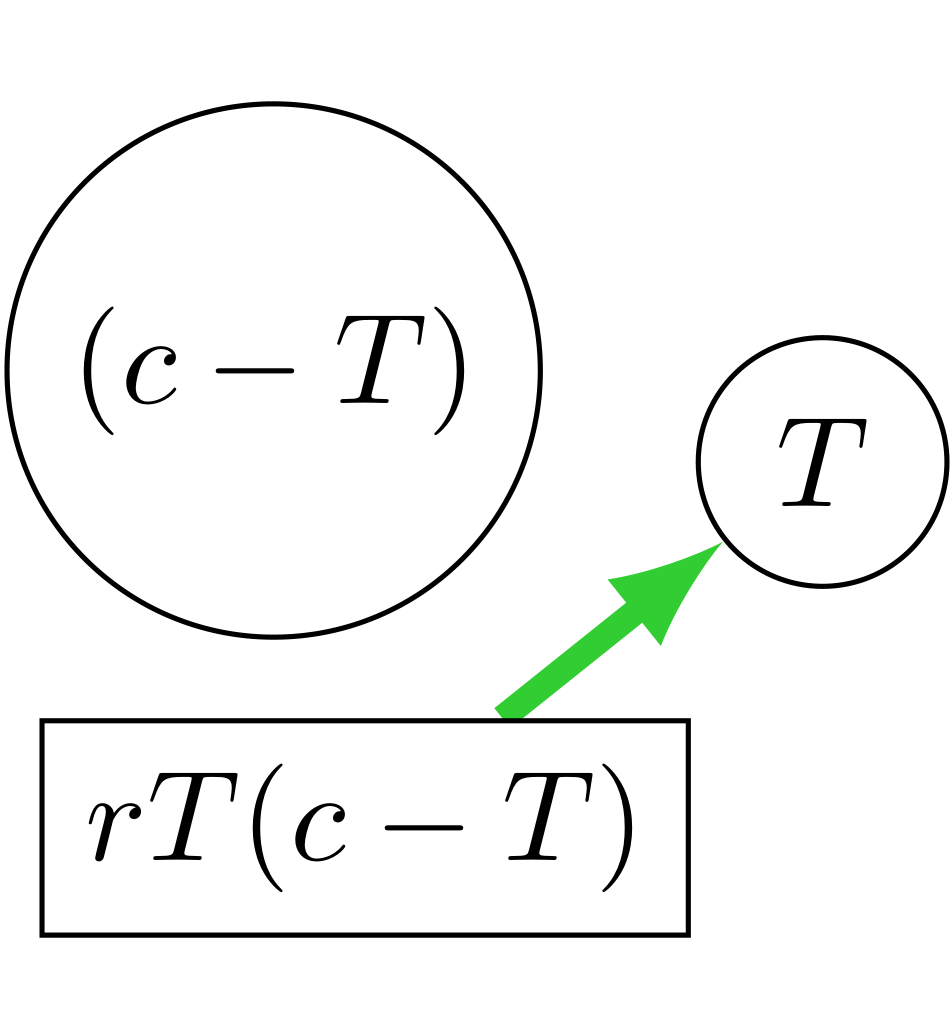}
  \captionof{figure}{}
  \label{fig:logen4}{}{}
\end{minipage}
\caption{We apply the modelling assumption at the last step between 
Figures \ref{fig:logen3} and \ref{fig:logen4} by cutting all incoming and outgoing edges from
the species $E$, then replacing all occurrences of $E$ with $(c-T)$.}
\end{figure}

Reading off the dynamical system from the only arrow left, we have:

$$\dot{T} = rT(c-T)$$

This is again the logistic equation, now with a capacity constant $c$ in place of $\frac{d}{r}$. What is notable here is that the same qualititative behaviour may be obtained from different starting assumptions, which is a challenge to qualitative modelling at large: the mere agreement of empirical data with a dynamical system may still leave the underlying processes unknown. While dynamical systems models may offer good predictions of behaviour of a system in a certain regime, their predictions may be frustrated by a transparent change of underlying processes, as there may be no evident way to transform the model to match the underlying change.

\subsection{(Unbounded) Catalysis}

Once each atomic process has been assigned a rate coefficient, we assume that their aggregate influence on the system is simply additive; framed differently, this is an assumption that the different processes are `disjoint', their relative prevalence only determined by their weights and the Law of Mass Action. However, the rate at which a particular process occurs may not be purely dependent the available input reactants; these rates may depend on the presence of catalysts and inhibitors. A catalyst (\emph{resp.} inhibitor) for a process is a species in the system that is neither consumed nor created by the process, but which the rate of the process depends positively (\emph{resp.} negatively) upon.\\

In what follows up through the Hill equation, we are chiefly concerned with modelling Catalysis and Inhibition `factored through' the Law of Mass Action; this reflects the underlying particle-like assumption that the rates of all processes are determined by just the Law of Mass Action and some modelling assumptions. What we do not permit is the possibility that the concentration of some species appears `naked' in a transition's rate coefficient while not being one of that transition's inputs.\\

The simplest form of catalysis modellable in our framework is an atomic process $A \otimes K \overset{r}{\multimap} B \otimes K$, where $A$ and $B$ are some species possibly with multitude, and $K$, the catalyst, is a single species. The net change in $K$ is 0, but the rate at which $A$ is converted into $B$ is proportional to the available amount of $K$.

\begin{figure}[h]
\centering
\begin{minipage}{.3\textwidth}
  \centering
  \includegraphics[width=0.8\linewidth]{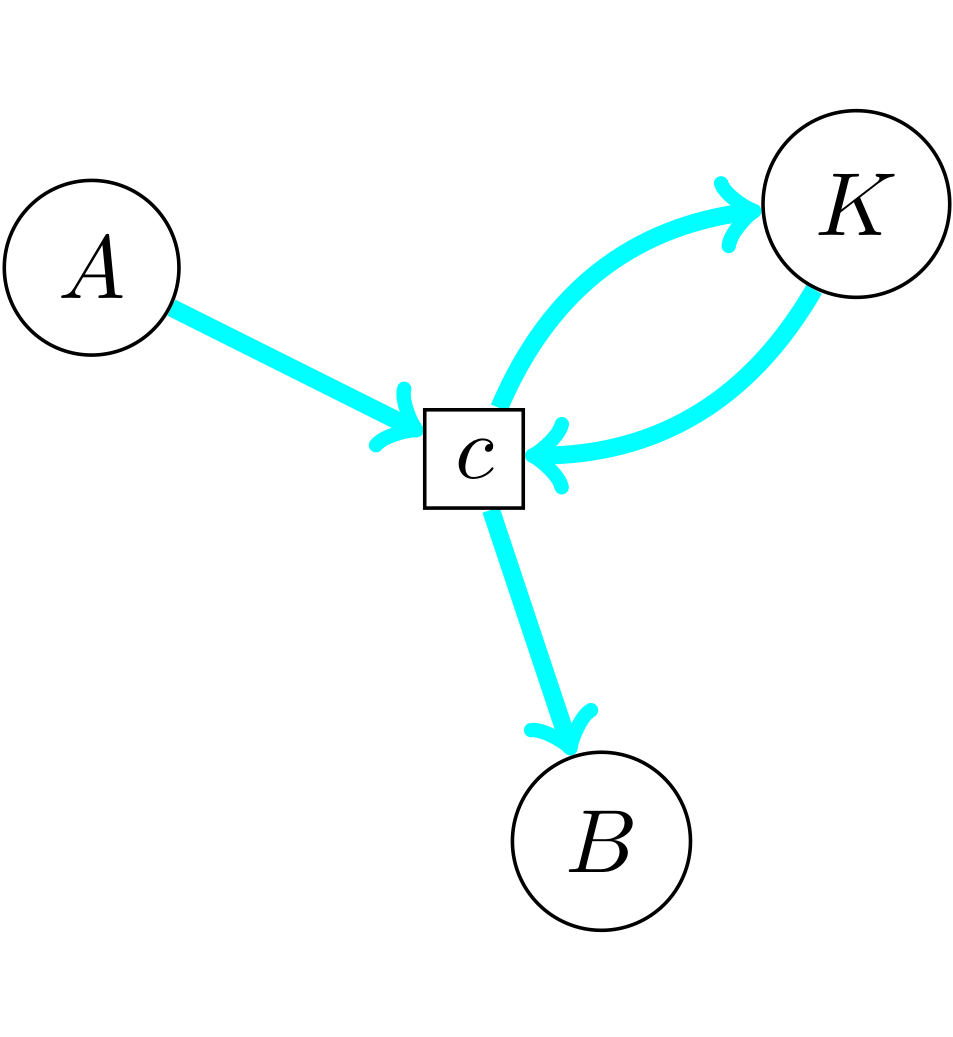}
  \captionof{figure}{}
  \label{fig:catunb}{}{}
\end{minipage}
\begin{minipage}{.3\textwidth}
  \centering
  \includegraphics[width=.8\linewidth]{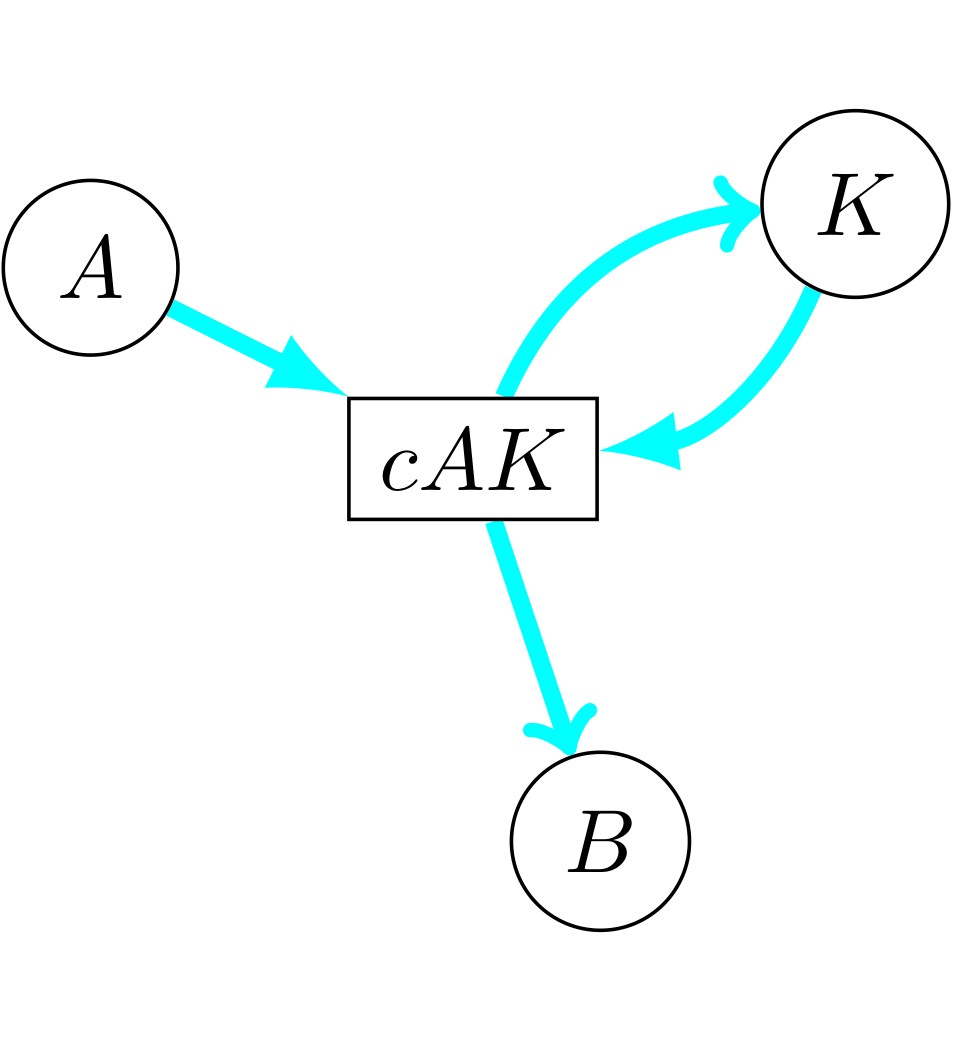}
  \captionof{figure}{}
  \label{fig:catunb2}{}{}
\end{minipage}
\begin{minipage}{.3\textwidth}
  \centering
  \includegraphics[width=.8\linewidth]{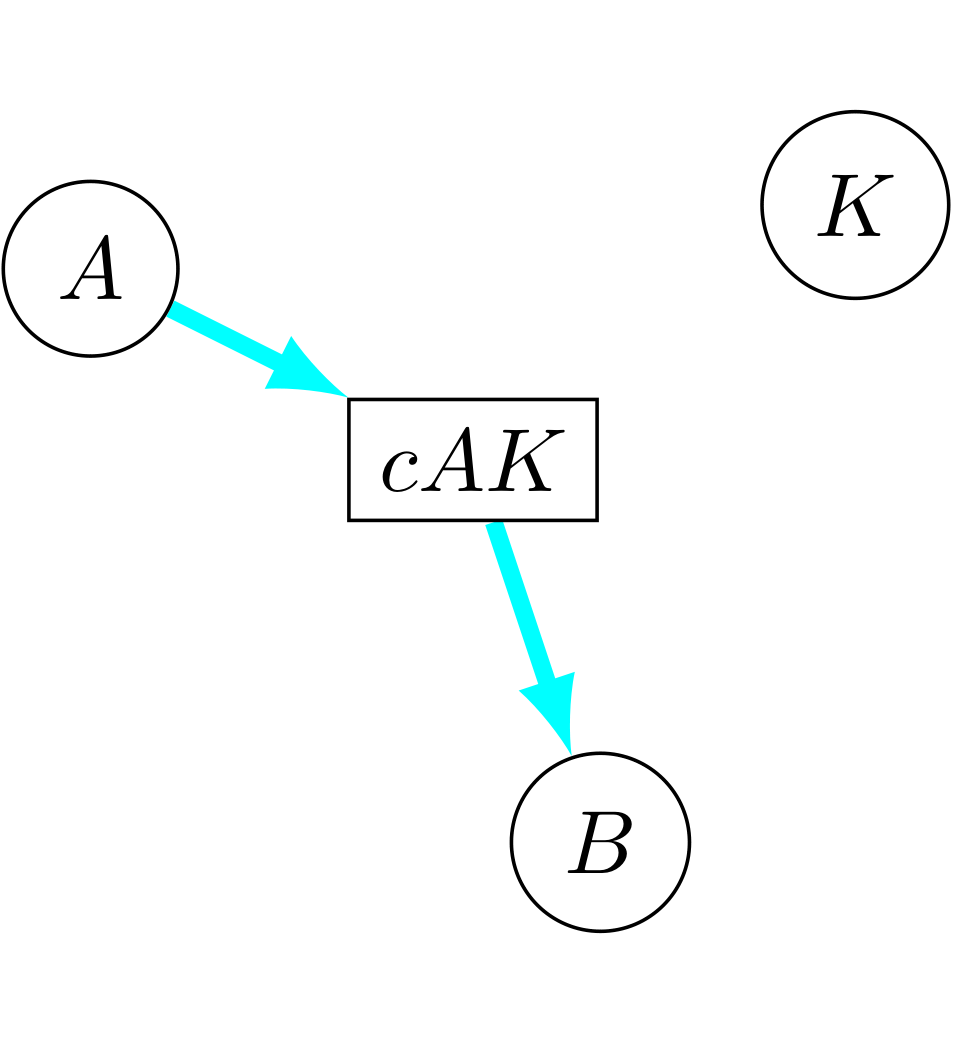}
  \captionof{figure}{}
  \label{fig:catunb3}{}{}
\end{minipage}
\caption{We see that the procedure of cancelling opposing edges allows us to recover a neat graphical derivation of the operational reading of catalysis: the catalyst $K$ appears to affect the rate of a reaction $A \multimap B$ from a distance, as a consequence of enjoying zero net change.}
\end{figure}

\subsection{Hill Equation}

Now some biochemistry. Say that one has a receptor protein $P$ that binds to a ligand $L$, resulting in a system with three species, $\{P_f, P_o, L\}$, where $P_f$ is a protein $P$ with a free active site, $P_o$ a protein with active site occupied by the ligand, and $L$ the ligand. The Hill equation relates the concentration of free ligands in a solution to the proportion $\frac{P_o}{P}$ of bound proteins, as a rectangular hyperbola $\mathbf{R} \mapsto [0,1)$. The form of the Hill equation is:

$$\theta = \frac{|L|^n}{k+|L|^n}$$

Where $\theta$ is the fraction of the receptor protein bound to the ligand, $|L|$ is the free unbound ligand concentration, $k$ is the ``apparent dissociation constant", and $n$ is the Hill coefficient: a free parameter.\\

We can reverse engineer the Hill equation from the following simple process specification. We consider there to be two processes in the system. The first is the binding of a ligand, which we assign an association coefficent $k_a$: $$P_f \otimes L \overset{k_a}{\multimap} P_o$$
The second is the dissociation of a bound ligand, which we assign a dissociation coefficient $k_d$:
$$P_o \overset{k_d}{\multimap} P_f \otimes L$$

We can see what this looks like in Figure \ref{fig:hill}. To obtain the Hill equation, we make two modelling assumptions. First, we assume that the amount of protein in the system is constant. Normalising, we can express this assumption via the assignment $P_f := (1-P_o)$, which we execute in Figure \ref{fig:hill3}. Secondly, we assume that the system is closed in such a way that $P_o \mapsto \mathbf{h}(L)$, \emph{for some} expression $\mathbf{h}$\footnote{we may as well assume that the field we're using is algebraically closed.}. We execute this assumption in Figure \ref{fig:hill4}. Thirdly -- and not a modelling assumption as we have conceived them -- is the assumption that the system is in equilibrium, which is tantamount to saying $\dot{L} = 0$ after the first two assumptions. We revisit this after the diagrams.

\begin{figure}[h]
\centering
\begin{minipage}{.45\textwidth}
  \centering
  \includegraphics[width=0.8\linewidth]{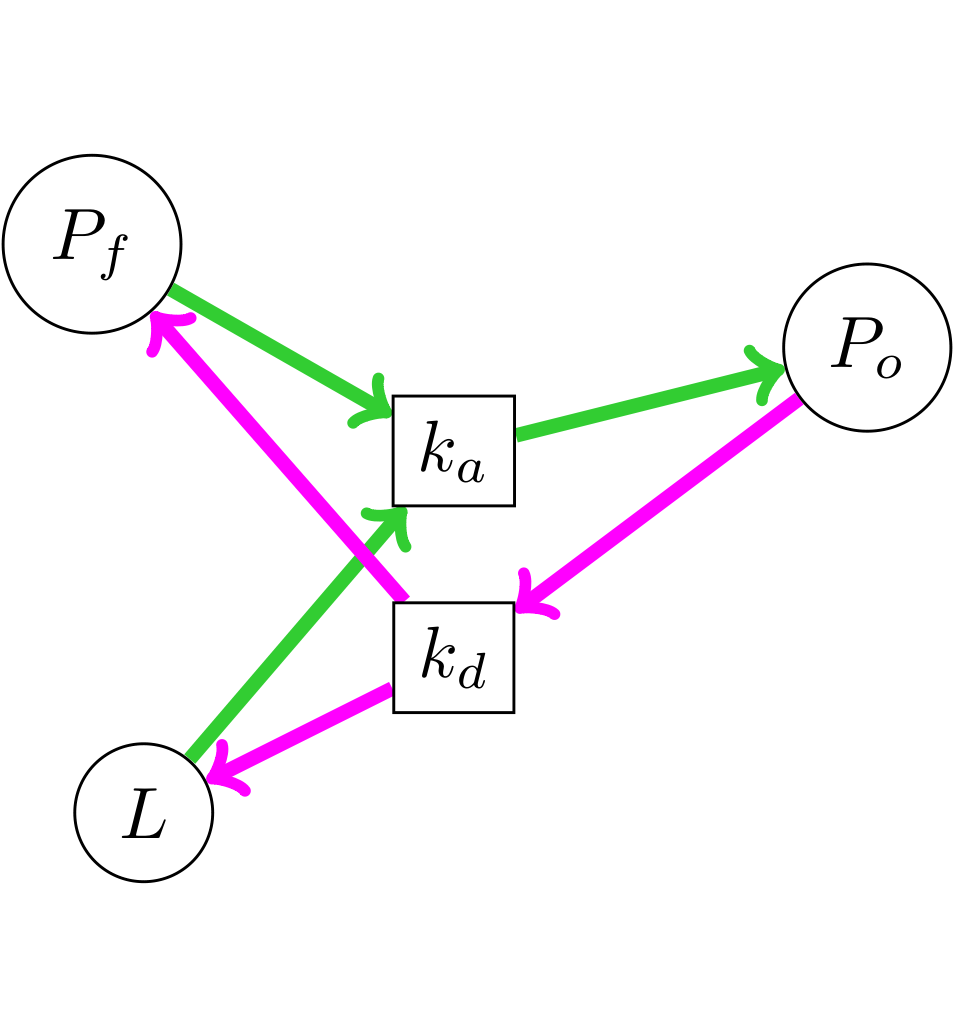}
  \captionof{figure}{}
  \label{fig:hill}{}{}
\end{minipage}
\begin{minipage}{.45\textwidth}
  \centering
  \includegraphics[width=0.8\linewidth]{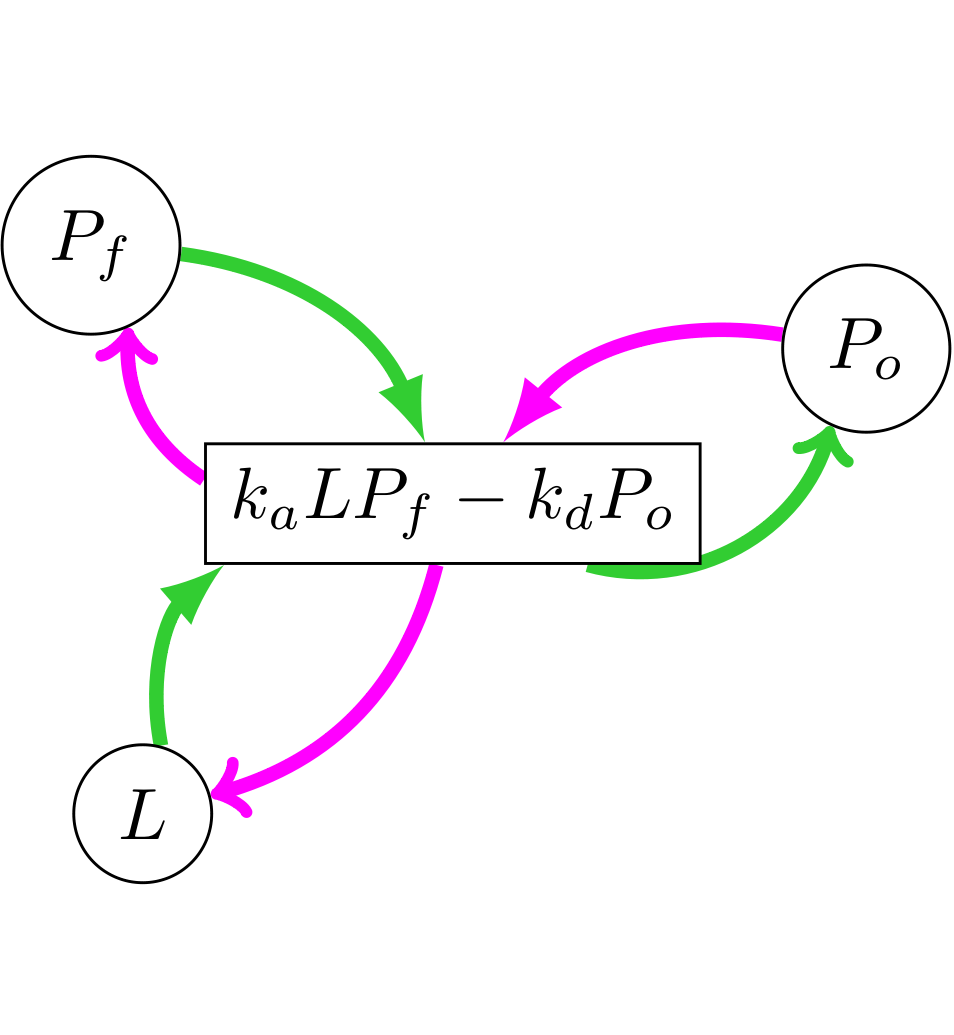}
  \captionof{figure}{}
  \label{fig:hill2}{}{}
\end{minipage}
\begin{minipage}{.45\textwidth}
  \centering
  \includegraphics[width=0.8\linewidth]{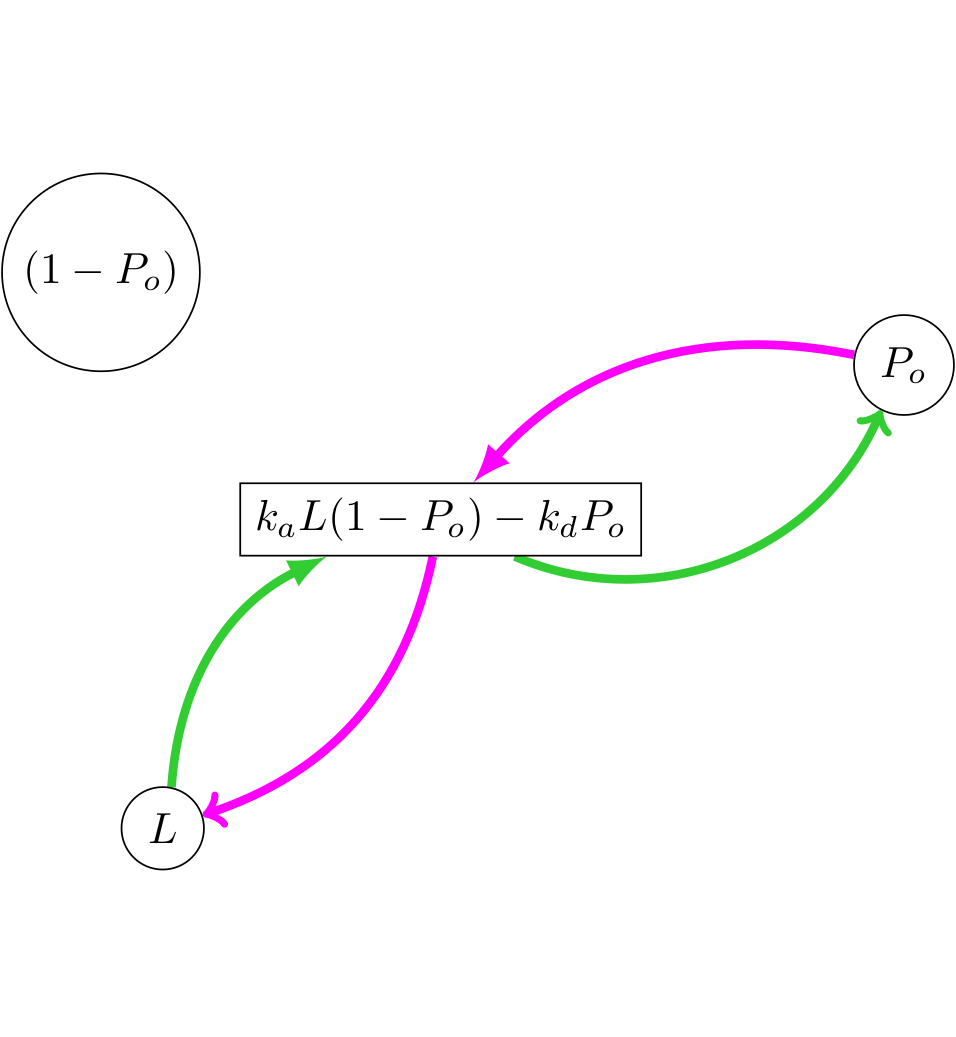}
  \captionof{figure}{}
  \label{fig:hill3}{}{}
\end{minipage}
\begin{minipage}{.45\textwidth}
  \centering
  \includegraphics[width=0.8\linewidth]{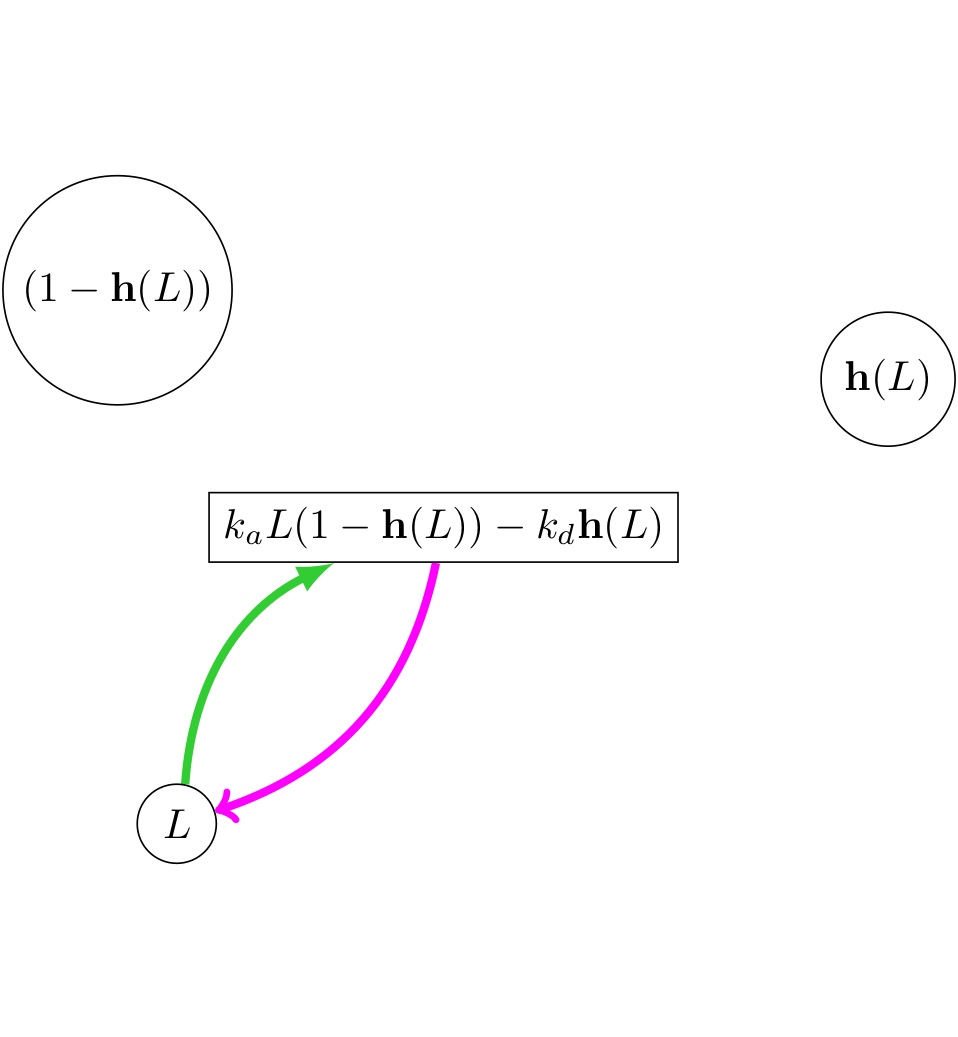}
  \captionof{figure}{}
  \label{fig:hill4}{}{}
\end{minipage}
\caption{We take a little liberty in Figure \ref{fig:hill2} and pre-emptively merge the processes, as we know we want to eliminate $P_f$ and $P_o$. We perform the modelling assumptions in steps. In Figure \ref{fig:hill3}, we replace $P_f$ with $(1-P_o)$, and cut edges. In Figure \ref{fig:hill4}, we replace $P_o$ with a placeholder expression in terms of $L$, and we cut edges.}
\end{figure}

So we have a dynamical system on just the ligand concentration:

$$\dot{L} = k_a L (1-\mathbf{h}(L)) - k_d \mathbf{h}(L)$$
The third assumption, which sets $\dot{L} = 0$, requires some non-diagrammatic mathematics.

\begin{align*}
&\implies &k_a L (1-\mathbf{h}(L)) = k_d \mathbf{h}(L)\\
&\implies &\mathbf{h}(L) = \frac{L}{\frac{k_d}{k_a}+L}
\end{align*}

We have recovered the Hill equation with apparent dissociation coefficient $\frac{k_d}{k_a}$ (the dissociation/association ratio) and Hill coefficient 1. We may recover other integer Hill coefficients\footnote{An aside on physical interpretations: a more reasonable alternative to considering the protein to have $n$ active sites is to suppose that \emph{in expectation}, the protein must encounter some number $n$ of ligands before a successful binding, possibly because the key-lock mechanism is quite complex, and not just any meeting between the protein and ligand results in a successful bond.} by amending our basic processes to require multiple free ligands as input/output: $P_f \otimes \bigotimes^n L \overset{k_a}{\multimap} P_o$ and $P_o \overset{k_d}{\multimap} P_f \otimes \bigotimes^n L$ by the same reasoning yields the Hill equation with Hill coefficient $n$.

\subsubsection{Inhibition}

Once we have done some algebra to obtain approximate solutions for petri nets, any time we see a subdiagram in a larger that looks like a system we have already solved for is amenable to simplification.\\

We present a toy model of inhibition via the Hill equation. $I$ behaves as an inhibitor of $A \multimap B$ when -- under some auxiliary assumptions -- $A$ is the input requirement of a process $A \multimap B$, and there are processes in the system which place $A$ in the role of free receptor, $A^*$ in the role of bound receptor, and $I$ as ligand.

The simplifying assumptions are as follows. First, we assume that the species $A^*$ and $I$ are in equilibrium, or close, in the reaction network, \emph{i.e.} $\dot{A^*} \approx 0$ and $\dot{I} \approx 0$: equilibrium conditions were necessary to justify the use of the Hill equation in our prior derivation. Secondly, we assume that the processes involving $A^*$ and $I$ happen on a much faster timescale than processes involving $A$, such that $\dot{A} \approx 0$. In conjunction with the first assumption that grants $\dot{A}^* \approx 0$, we may have $\dot{A} + \dot{A^*} \approx 0$, so we may treat the total amount of $A$ in the system -- whether bound or free -- as approximately constant, and thus normalisable: \emph{i.e.} $A + A^* \approx 1$ at all times, so we have access to the $A \mapsto (1-A^*)$ assumption mirroring our derivation for the Hill equation.

\begin{figure}[h]
\centering
\begin{minipage}{.3\textwidth}
  \centering
  \includegraphics[width=0.8\linewidth]{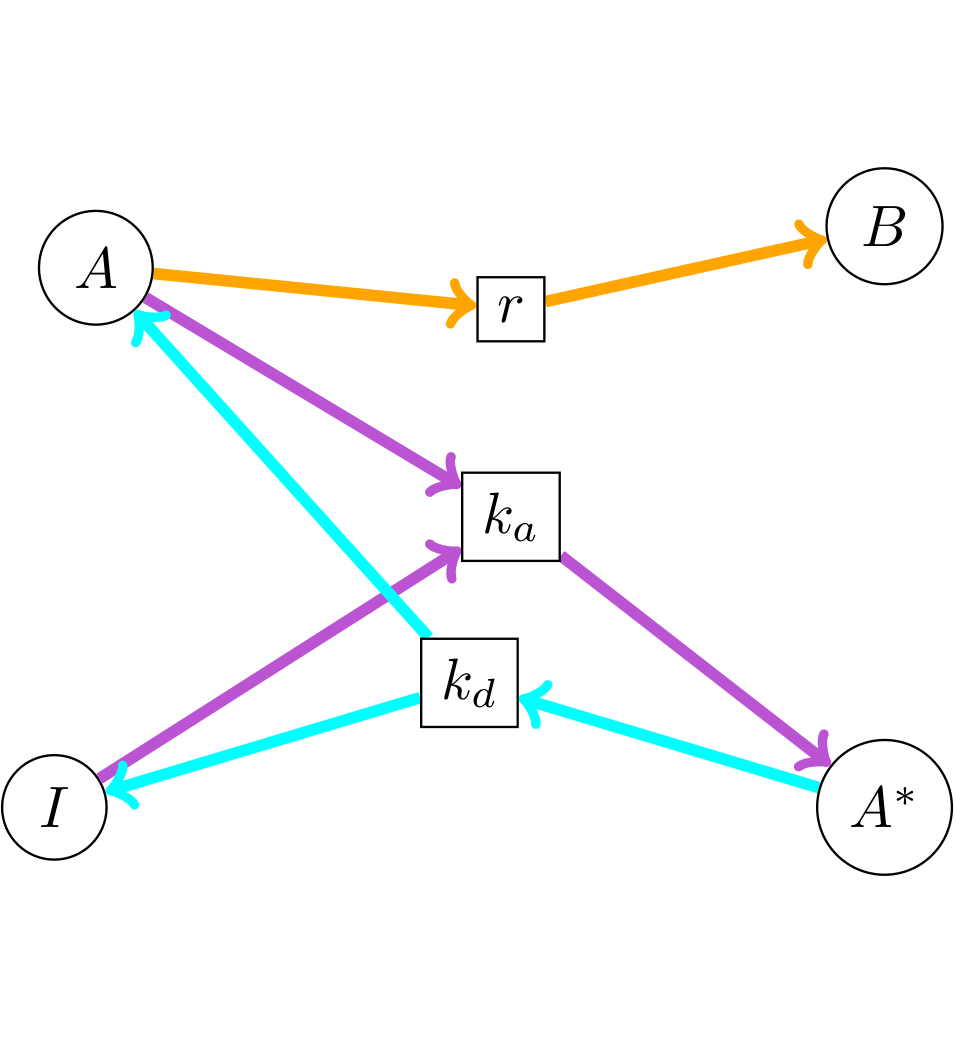}
  \captionof{figure}{}
  \label{fig:hillinhib}{}{}
\end{minipage}
\begin{minipage}{.3\textwidth}
  \centering
  \includegraphics[width=.8\linewidth]{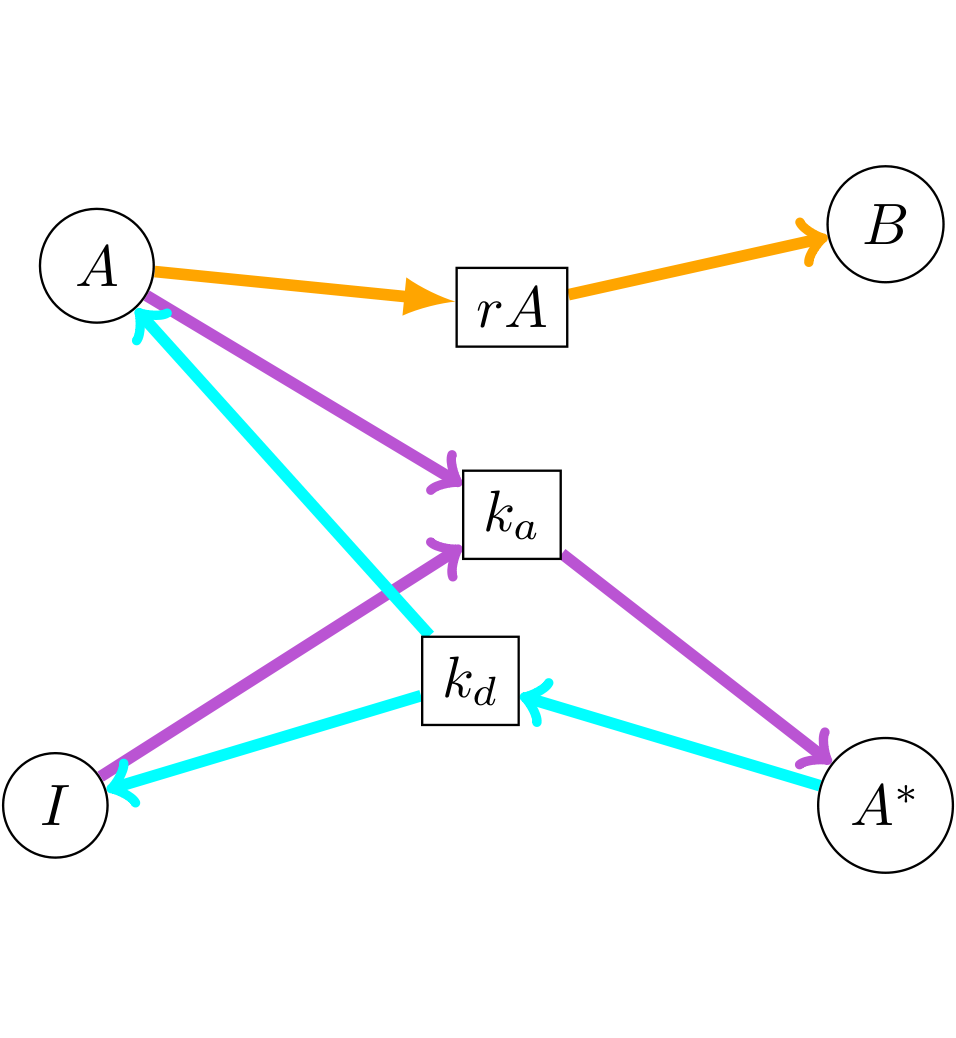}
  \captionof{figure}{}
  \label{fig:hillinhib2}{}{}
\end{minipage}
\begin{minipage}{.3\textwidth}
  \centering
  \includegraphics[width=.8\linewidth]{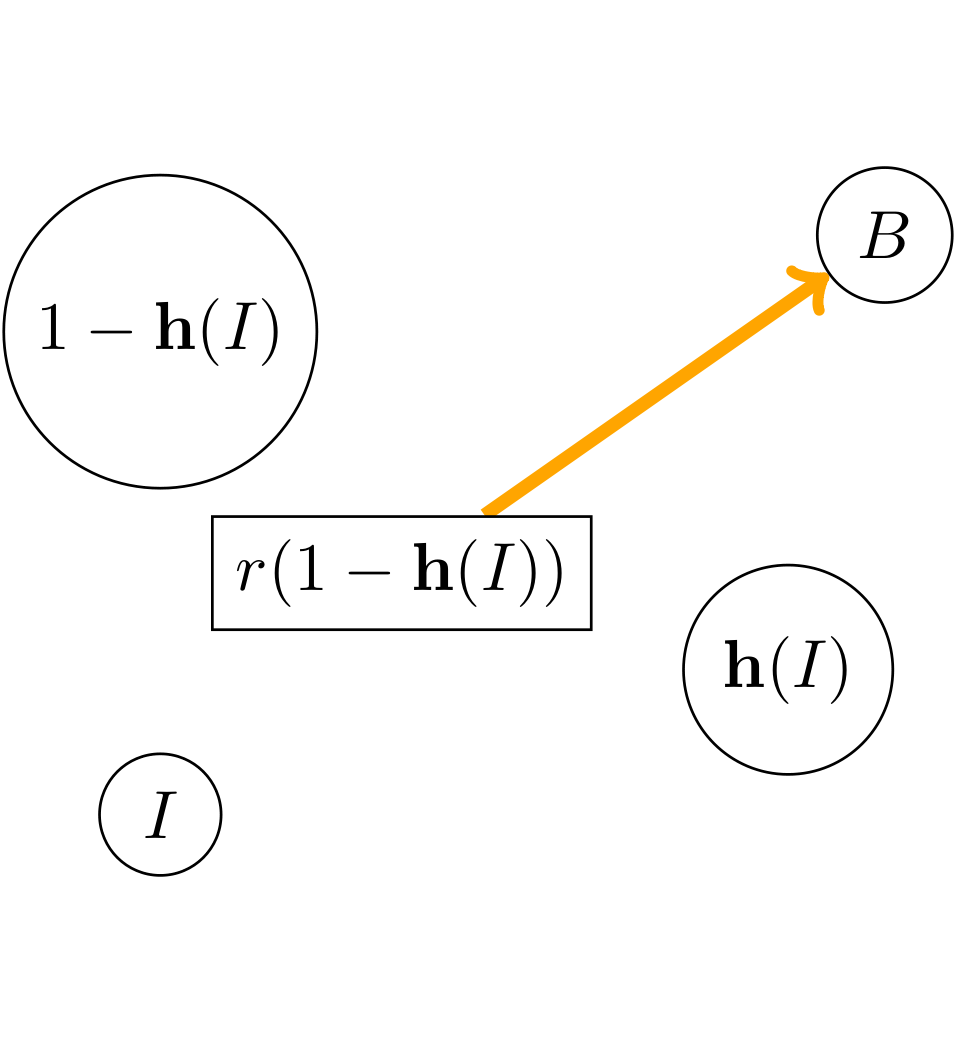}
  \captionof{figure}{}
  \label{fig:hillinhib3}{}{}
\end{minipage}
\caption{We note that inhibition is quite a complex phenomenon, compared to catalysis. Technically speaking, the Hill function modelling assumption is inappropriate here, as $A$ is being consumed to produce $B$, but this can be amended by assuming $A \overset{r}{\multimap} A \otimes B$, without compromising the intention of the example.}
\end{figure}

The assumptions we have made are only necessary to employ Hill equations, not for achieving an inhibitory effect: even without the assumptions, $I$ retains an inhibitory effect, which is evident from a qualitative examination of the reaction schema. As the quantity of $I$ increases, by Law of Mass Action, more $A$ are bound into a form $A^*$ and made unavailable for $A \multimap B$. The fundamental principle is that the processes $A \multimap B$ and $A \otimes I \multimap A^*$ are in competition for the same finite pool of $A$.

\subsection{Allee Effect}

The Allee effect \cite{allee_studies_1932} in ecology occurs where a population exhibits bistable behaviour: there are three equilibria, which are, in increasing order, $T = 0$, $T = T_{\text{crit}}$, and $T = T_{\text{max}}$. Further, $0$ and $T_{\text{max}}$ are stable equilibria, and $T_{\text{crit}}$ is unstable. The qualitative effect is that when the population is above the critical size $T_{\text{crit}}$, it will tend to grow until reaching $T_{\text{max}}$. Below the critical threshold, the population will tend to extinction.\\

In practice, the Allee effect can arise when things require other things around to support reproduction \cite{davis_pollen_2004} or when cooperative behaviour in things is disabled when there are too few things around \cite{berec_multiple_2007}. We will show how the qualitative characteristics of the Allee effect may be recovered in our framework, in turn shedding light on the structure of the basic processes from which the effect may arise.\\

We should remark that Allee Effects are typically modelled by cubic first-order differential equations in one variable, of the form $\dot{T} = T(T_{\text{crit}}-T)(T-T_{\text{max}})$. While this does provide the correct qualitative behaviour in the system, this `equilibrium-pasting' approach is incompatible with knowledge of what underlying processes might already constitute the system in question. Our approach starts from knowledge of basic processes to \emph{recover} the Allee effect.\\

Let us suppose we have the species $T$, $E$, and $TT$. The first two are the familiar `things' and `energy', and the third is a new species which we consider `pair-bonded things', that behave as the parents of a familial unit. Suppose that there are association and dissociation processes $T \otimes T \multimap TT$ and $TT \multimap T \otimes T$ with rates $k_a$ and $k_d$ respectively. The remainder of the setup is similar to our approach for modelling logistic growth with the finite energy assumption. The birth process is typed $TT \otimes E \multimap TT \otimes T$; parents may consume some energy from the system and create a child thing. We also amend the death process, now typed $T \multimap E$; when a thing dies, it returns its essence to the earth, or is cannibalised.\\

In this setup, we model the fraction of pair-bonded things $TT$ as a Hill function $\mathbf{h}(T)$, so we capture the total number of $TT$ with the modelling assumption $TT \mapsto T\mathbf{h}(T)$. We keep the finite energy assumption $E \mapsto (c - T)$ from before. We will also adopt the facilitating assumption that association and dissociation rates are so large compared to birth and death rates that the association-dissociation subsystem is essentially equilibrated\footnote{Living appears to occur more frequently than being born or dying.}, which is not a modelling assumption, but will allow us to cut two more edges in Figure \ref{fig:allee4}.

\begin{figure}[h]
\centering
\begin{minipage}{.45\textwidth}
  \centering
  \includegraphics[width=0.8\linewidth]{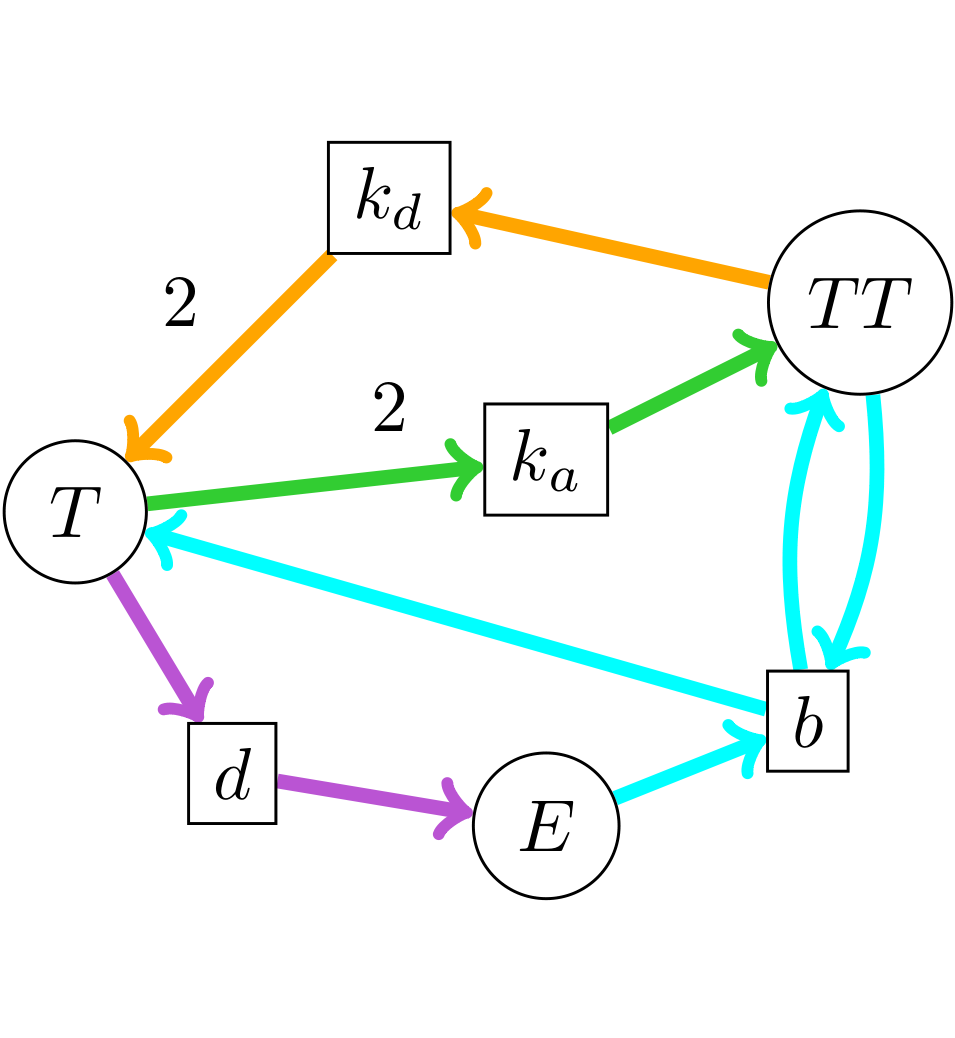}
  \captionof{figure}{}
  \label{fig:allee}{}{}
\end{minipage}
\begin{minipage}{.45\textwidth}
  \centering
  \includegraphics[width=0.8\linewidth]{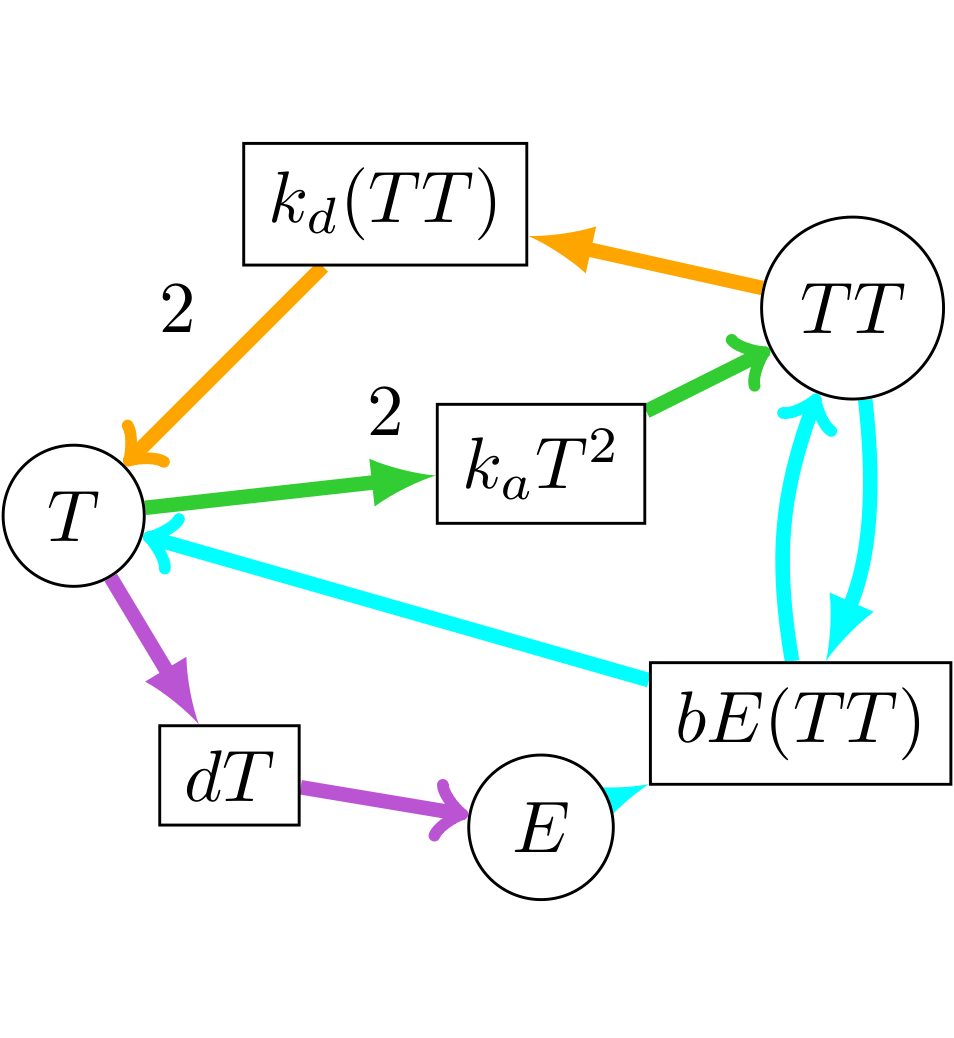}
  \captionof{figure}{}
  \label{fig:allee2}{}{}
\end{minipage}
\begin{minipage}{.45\textwidth}
  \centering
  \includegraphics[width=0.8\linewidth]{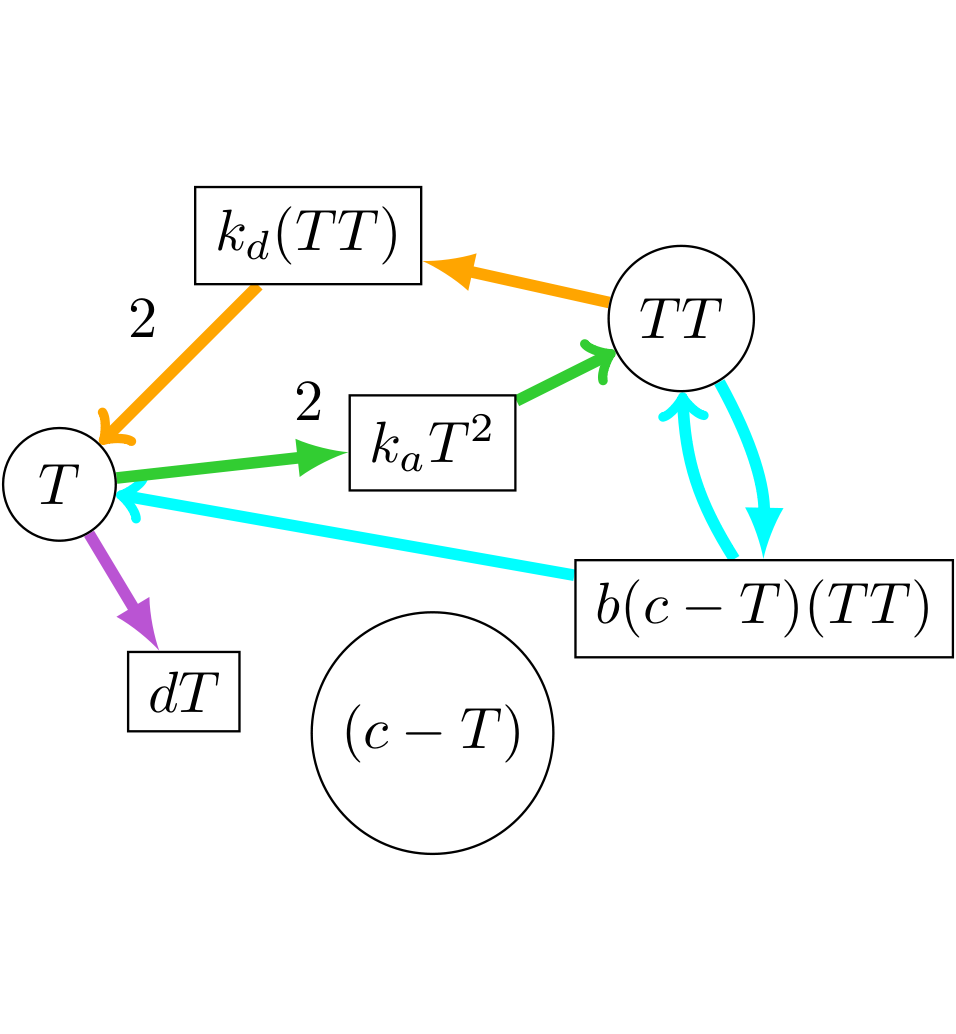}
  \captionof{figure}{}
  \label{fig:allee3}{}{}
\end{minipage}
\begin{minipage}{.45\textwidth}
  \centering
  \includegraphics[width=0.8\linewidth]{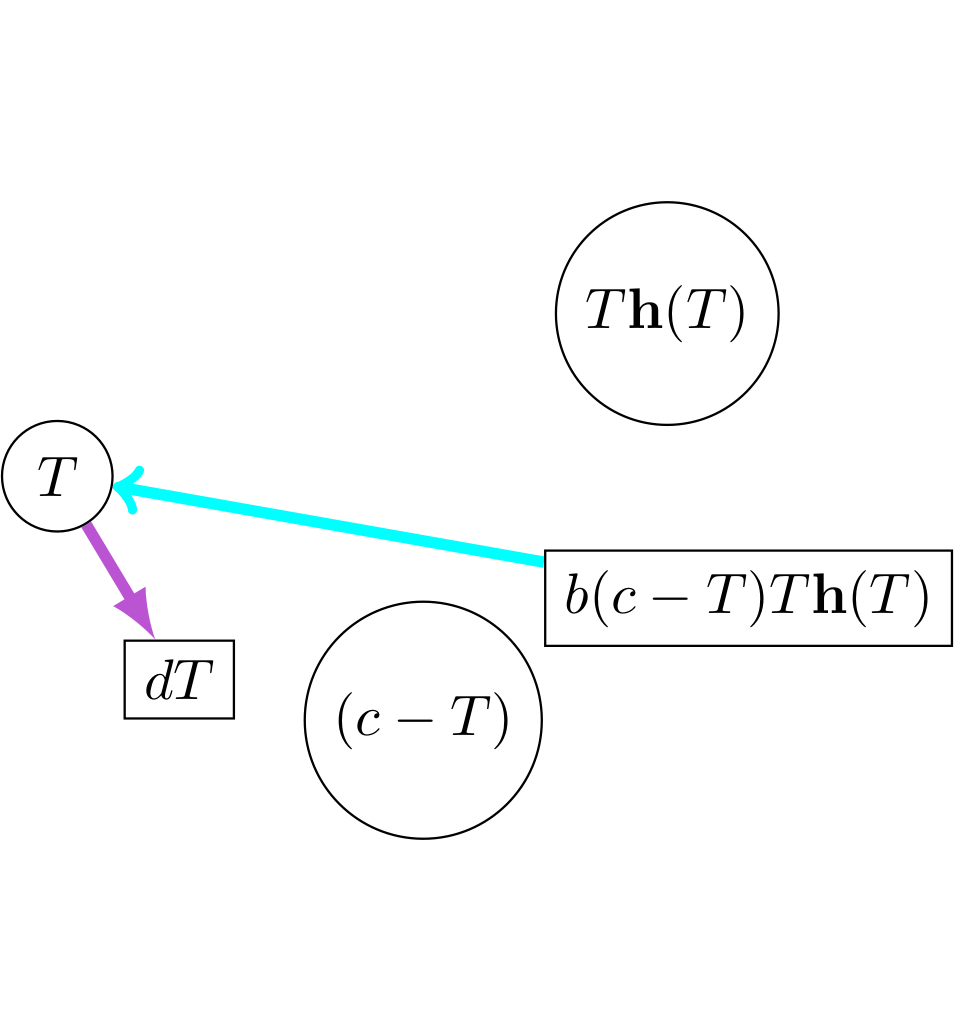}
  \captionof{figure}{}
  \label{fig:allee4}{}{}
\end{minipage}
\caption{In the transition from the first to second figure, we calculate mass-respecting rates. In the transition from second to third, we execute $E \mapsto (c-T)$, and in the final transition, we execute $TT \mapsto T\mathbf{h}(T)$. We also assume $k_a$ and $k_d$ equilibrate the outgoing and incoming edges from $T$ due to association and dissociation, so we remove those too.}
\end{figure}{}

We can read off the resultant dynamical system on the variable $T$ as:

$$\dot{T} = b \overbrace{(c-T)}^{\text{finite energy}} \underbrace{T  \mathbf{h}(T)}_{\text{from }TT} \  - \  dT$$

\begin{figure}[h]
\centering
\includegraphics[width=0.7\linewidth]{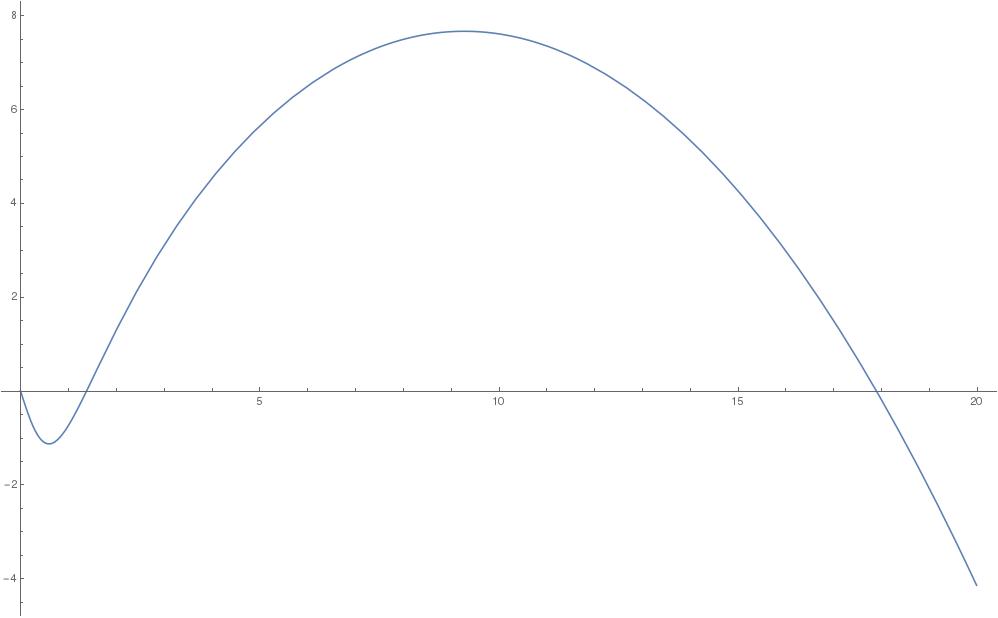}
\caption{Plot of $\dot{T}$ against $T$ displaying Allee effect, with parameters $\text{b} \mapsto 0.1$, $\text{d} \mapsto 3.2$, and Hill coefficient $2$.}
\end{figure}

\clearpage

\section{Case Study: HES1 Gene Regulation Network}

In this section we demonstrate how we may recover quantitative models from a real gene transcription network that qualititatively match empirical data under control and intervention conditions. We consider the HES1 transcription network, which was demonstrated by \cite{hirata_oscillatory_2002} to exhibit oscillatory behaviour due to the presence of a negative feedback loop. Following a petri-esque schematic obtained from \cite{sturrock_spatial_2013}, the HES1 gene network is understood to consist of the following processes:

\begin{figure}[h]
\centering
\includegraphics[width=0.5\linewidth]{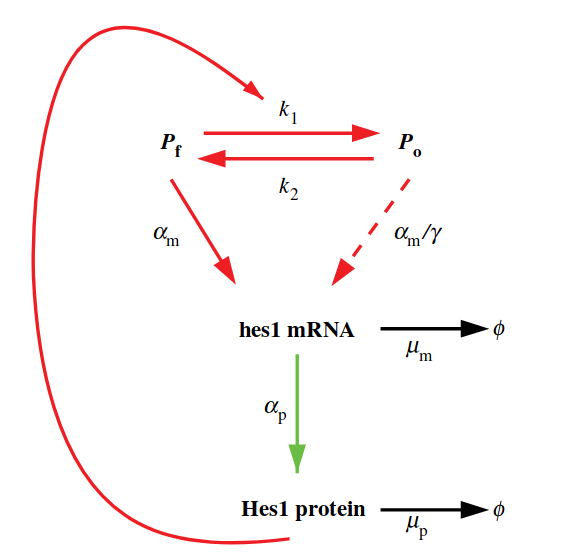}
\caption{(Picture and caption from \cite{sturrock_spatial_2013}, which we have truncated at the end to omit discussion of spatial aspects, so we may ignore the colours of the arrows): The negative feedback loop in the Hes1 gene reaction network (GRN). When the promoter site is free, hes1 mRNA is transcribed at its maximal rate (m). hes1 mRNA then produces Hes1 protein via the process of translation. Hes1 protein occupies the promoter and represses the transcription of its own mRNA. The occupied promoter site is still able to produce hes1 mRNA, but at a significantly reduced rate.}
\end{figure}

\begin{itemize}
\item{mRNA transribes/synthesises proteins without being consumed $$M \overset{t}{\multimap} M \otimes P$$}
\item{the protein HES1 binds to its own promoter P \begin{align*} H \otimes P_f \overset{k_a}{\multimap} P_o\\ P_o \overset{k_d}{\multimap} H \otimes P_f \end{align*}}
\item{P makes more HES1 mRNA. It makes it more effectively when free than when bound to the HES1 protein. \begin{align*} P_f \overset{m}{\multimap} P_f \otimes M \\ P_o \overset{\gamma \times m}{\multimap} P_f \otimes M \ (0<\gamma<1) \end{align*}

\underline{Our modelling assumption:} the quantity P is constant, so we can abstract it away with a Hill equation. $P_f \mapsto (1-\mathbf{h}(H))$, $P_o \mapsto \mathbf{h}(H)$.}
\item{mRNA and HES1 are constantly being reduced by cell garbage-collection. Hirata et al. calculate half-lives of $24.1 \pm 1.7$ minutes and $22.2 \pm 3.1$ minutes respectively.
\begin{itemize}
\item{In Hirata et als' supplementary materials, they provide the differential equation model:

\begin{align*}
&\mathbf{S}={H,M,P}\\
&\dot{H} =(-0.022) HP + (0.3) M - (0.031)H\\
&\dot{M} = (-0.028)M + \frac{(0.5)}{1+H^2} \\
&\dot{P} =  (-0.022)HP + \frac{20}{1+H^2} - (0.3)P\\
\end{align*}

The authors make explicit in their supplementary material that the $(0.031)$ and $(0.028)$ that govern exponential decay of $H$ and $M$ are obtained from their measured half-lives. $(0.031)$ and $(-0.028)$ are (approximate) solutions for $(1-x)^{24.1} = 1/2$ and $(1-x)^{22.2} = 1/2$ respectively, so in numerical simulations of Hirata \emph{et. al}'s model we chose their normalised timestep of 1 minute.
}
\item{Following their half-life measurements and minute-timestep, we fix the same rates for our garbage collection processes. \begin{align*} H \overset{0.031}{\multimap} 0 \\ M \overset{0.028}{\multimap} 0 \end{align*}}
\end{itemize}}
\end{itemize}

The model we obtain by our method is:

\begin{align*}
&\mathbf{S} = \{H,M\}\\
&\dot{H} = k_a H (1 - \mathbf{h}(H)) + k_d \mathbf{h}(H) + tM - (0.031)H\\
&\dot{M} = m (1 - \mathbf{h}(H)) + m \gamma \mathbf{h}(H) - (0.028)M
\end{align*}

With free parameters $\{k_a,k_d,m,\gamma,t,\mathbf{h} (\text{hill coefficient})\}$, where all values are positive, and in addition, $0 < \gamma < 1$.

\begin{figure}[h]
\centering
\begin{minipage}{.45\textwidth}
  \centering
  \includegraphics[width=0.8\linewidth]{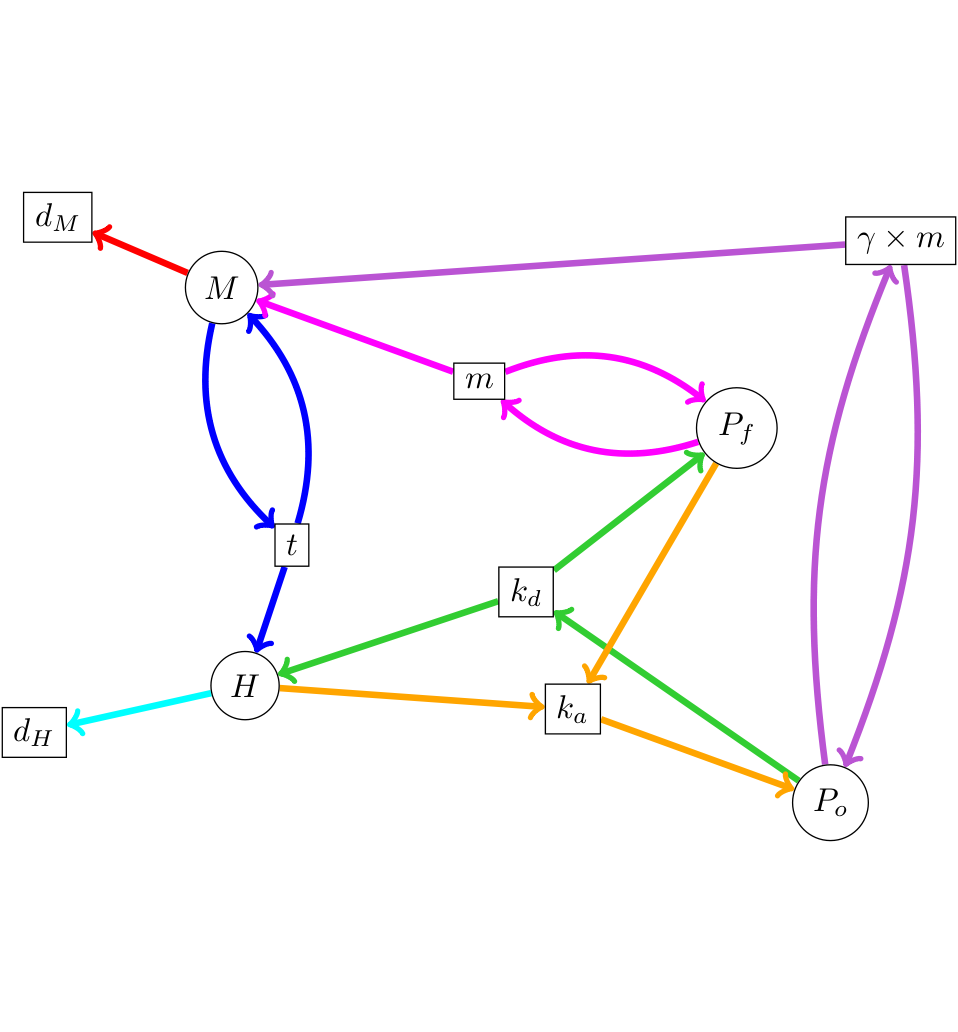}
  \captionof{figure}{}
  \label{fig:hes}{}{}
\end{minipage}
\begin{minipage}{.45\textwidth}
  \centering
  \includegraphics[width=0.8\linewidth]{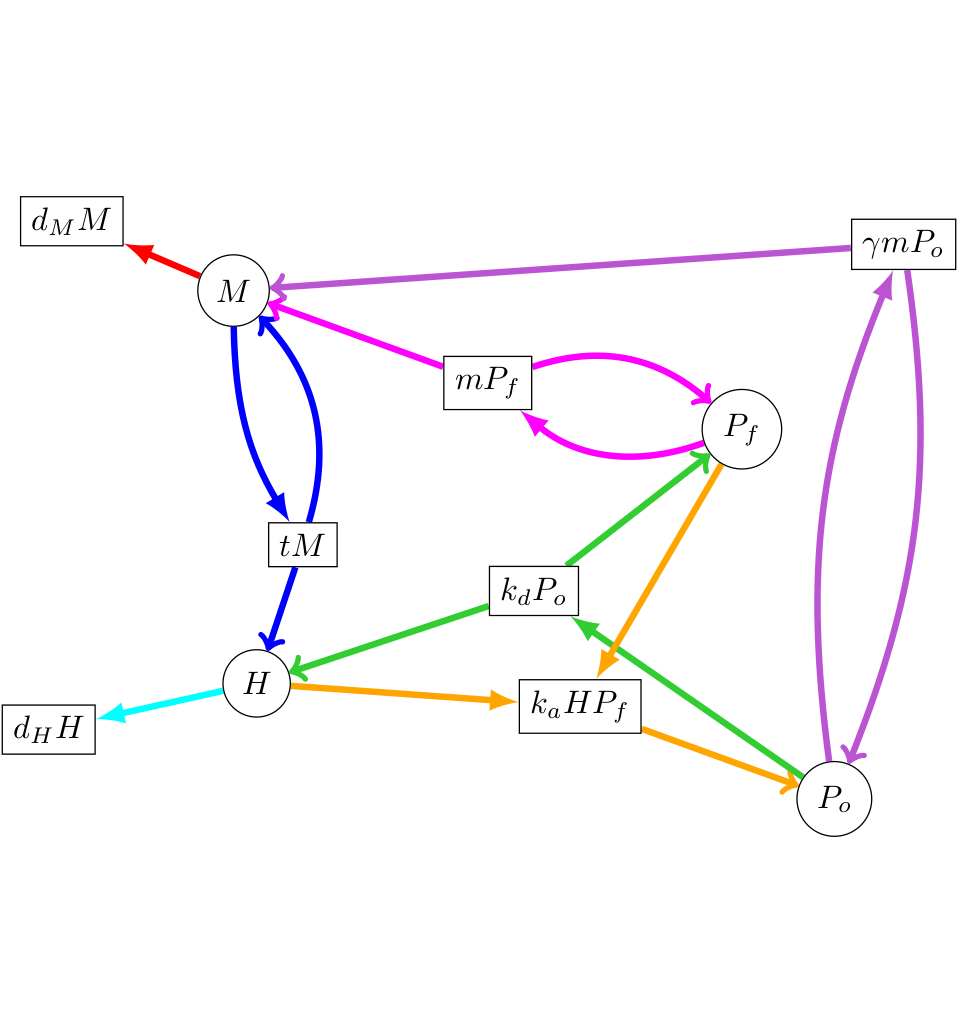}
  \captionof{figure}{}
  \label{fig:hes2}{}{}
\end{minipage}
\begin{minipage}{.45\textwidth}
  \centering
  \includegraphics[width=0.8\linewidth]{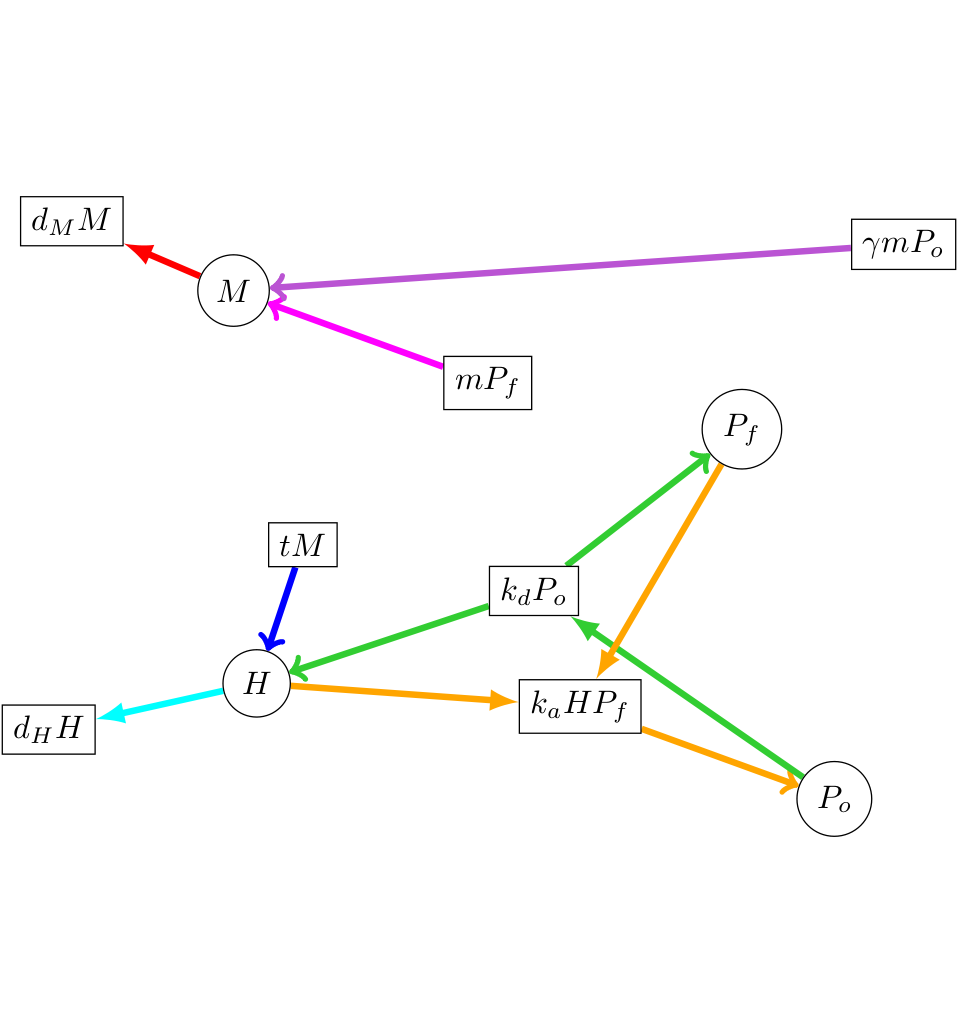}
  \captionof{figure}{}
  \label{fig:hes3}{}{}
\end{minipage}
\begin{minipage}{.45\textwidth}
  \centering
  \includegraphics[width=0.8\linewidth]{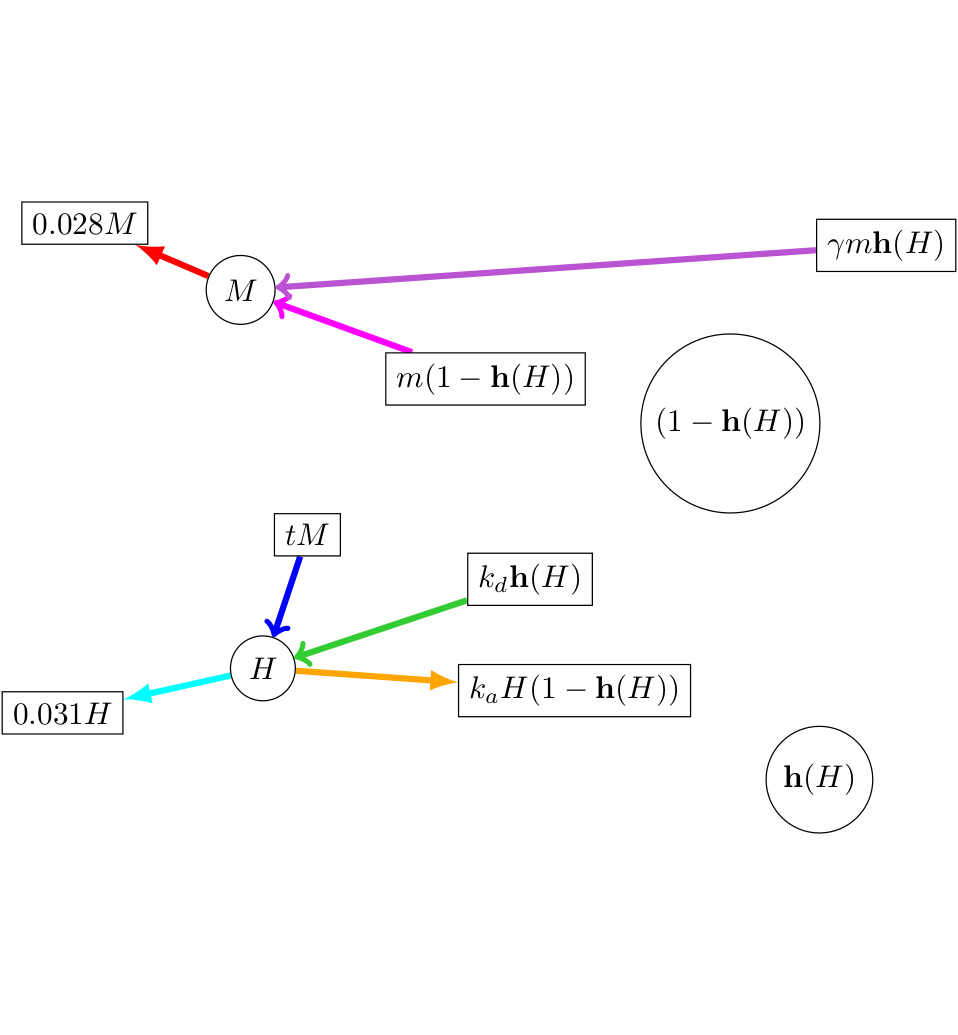}
  \captionof{figure}{}
  \label{fig:hes4}{}{}
\end{minipage}
\caption{We proceed through the manipulations as we have done before: calculating mass respecting rates, then cancelling opposing edges. In the last step, we apply the modelling assumptions $P_f \mapsto (1-\mathbf{h}(H))$ and $P_o \mapsto \mathbf{h}(H)$, cutting edges incident to $P_f$ and $P_o$. We also instantiate $d_H$ and $d_M$ with the empirically obtained values.}
\end{figure}

\clearpage

\subsection{Doing Computer Science}

Where empirical data is available, we can instantiate unknown rate constants using gradient descent. We could not access raw numerical data, but we were able to transcribe Hirata \emph{et. al}'s observations by hand-and-ruler from their graphs. We used a two-sided linear loss function around the midpoints of the error bars, which also included additional multiplicative penalty factors for large jumps in value (to prevent bifurcations in the recurrence relation, as we performed numerical simulation), non-periodic or eventually constant behaviour, deviation from period 120, and maximum values of 9 amd 4 for the protein and mRNA models respectively. Our resulting model (Figure \ref{ourmodel}) is qualitatively indistinguishable from that of Hirata \emph{et al} (Figure \ref{hirataoriginal})\footnote{neither are particularly good at threading the observed data.}.

\begin{figure}[h]
\centering
\includegraphics[width=0.6\linewidth]{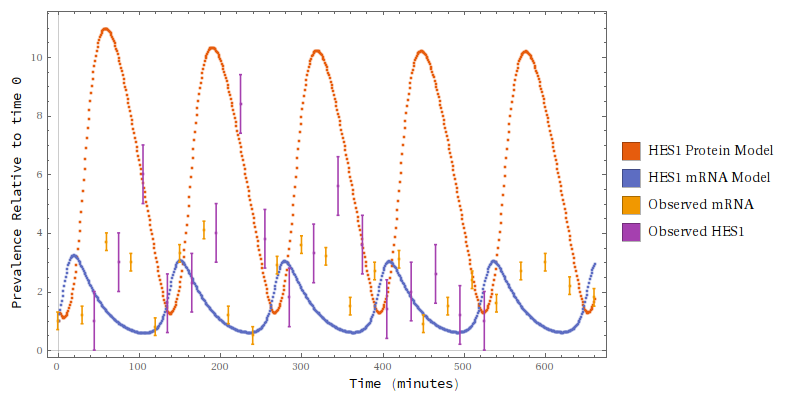}
\caption{Hirata \emph{et al}'s differential equation model. Their units for measuring change in HES1 protein and mRNA are normalised against the first observation in the time series, due to their methodology of quantifying their measurements. Their series of measurements of the two quantities were not performed at concurrent times (Hirta et. al fig 1A and 1B). The mRNA measurements were taken at thirty minute intervals starting at $t_0$, whereas the protein measurements were taken at thirty minute intervals starting at 45. We transcribed the data by ruler measurements and increased the error bars to reflect that. We assumed an initial state vector $\{1,1,1\}$ for the authors at time 0.}
\label{hirataoriginal}
\end{figure}

\begin{figure}[h]
\centering
\includegraphics[width=0.6\linewidth]{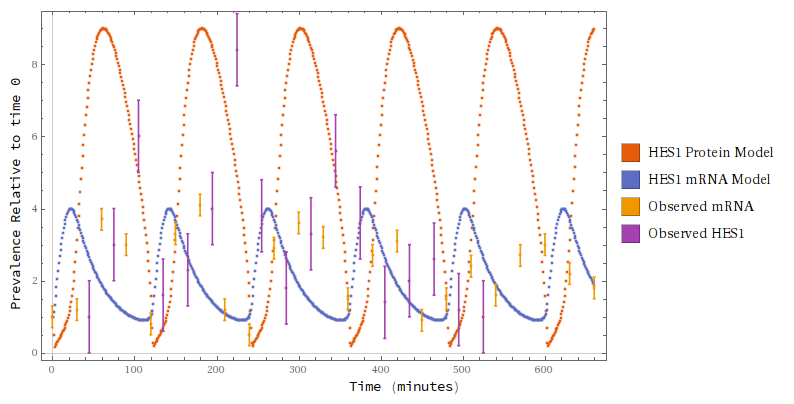}
\caption{Our model with fitted parameters $k_a \mapsto 1.35213,
k_d \mapsto 0.00969532,
m \mapsto 0.294401,
\gamma \mapsto 0.0385978,
t \mapsto 0.20275,
\mathbf{h} (\text{Hill equation power}) \mapsto 2.16271$.}
\label{ourmodel}
\end{figure}

\clearpage

\subsection{Doing Science}

Now we can do something that Hirata \emph{et. al} did not do, and to the best of our knowledge, has not been done before: a principled interventional evaluation of dynamical systems models with respect to empirical data. Petri nets are prescriptive frameworks for systems, while dynamical systems models are descriptive. Bridging the two by Law of Mass Action gives us the best of both worlds.\\

Hirata \emph{et al} collected data from two intervention conditions. The first is equivalent to a suppression of HES1 protein garbage collection. The second is equivalent to a suppression of the gene transcription process $M \multimap M \otimes H$ alongside a suppression of the mRNA garbage collection process $M \multimap 0$. Both of these intervention conditions target transitions in our initial petri net, so we can test the correctness of our fitted model with respect to data collected in the intervention conditions: if we can obtain a qualitative match to empirical data by reducing the rate constants of the suppressed transitions, then we at least know that our model is not obviously wrong.

\begin{figure}[h]
\centering
\begin{minipage}{\textwidth}
  \centering
  \includegraphics[width=0.3\linewidth]{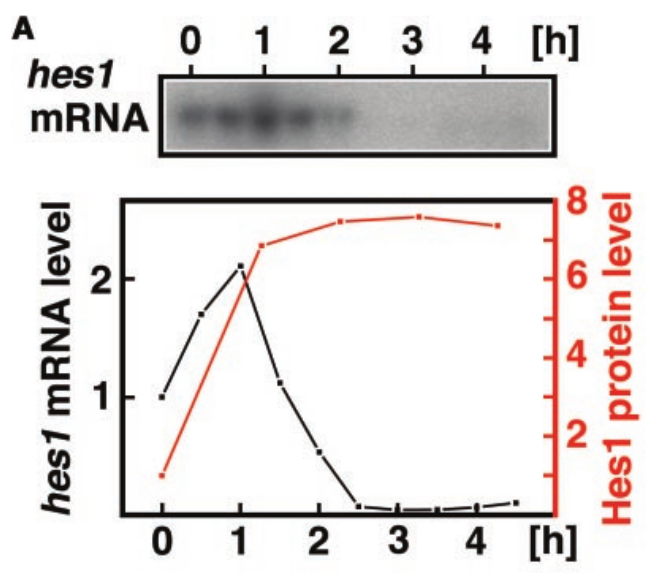}
\end{minipage}
\begin{minipage}{\textwidth}
  \centering
  \includegraphics[width=\linewidth]{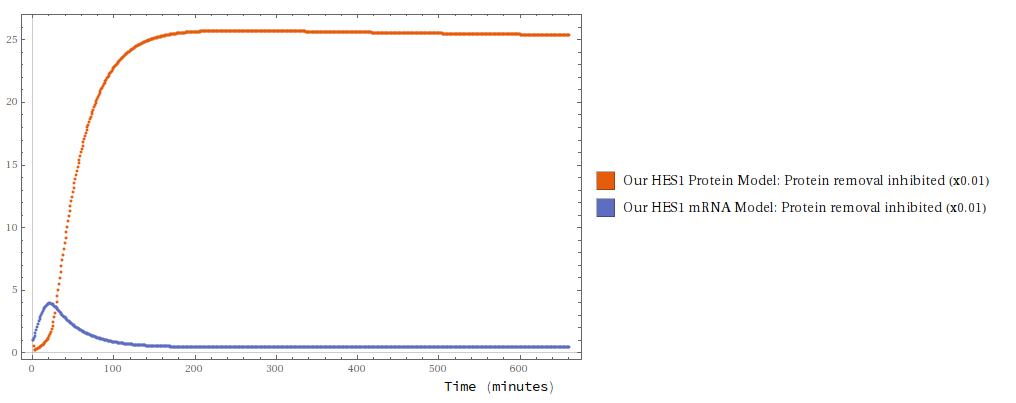}
\end{minipage}
\caption{Above: Hirata's empirical data obtained by intervention by inclusion of MG132, which inhibits protein destruction. Modelling protein destruction inhibition in our model. Below: Our HES1 Protein values are off by about a factor of 3, but the shapes look right.}
\end{figure}

\begin{figure}[h]
\centering
\begin{minipage}{\textwidth}
  \centering
  \includegraphics[width=0.3\linewidth]{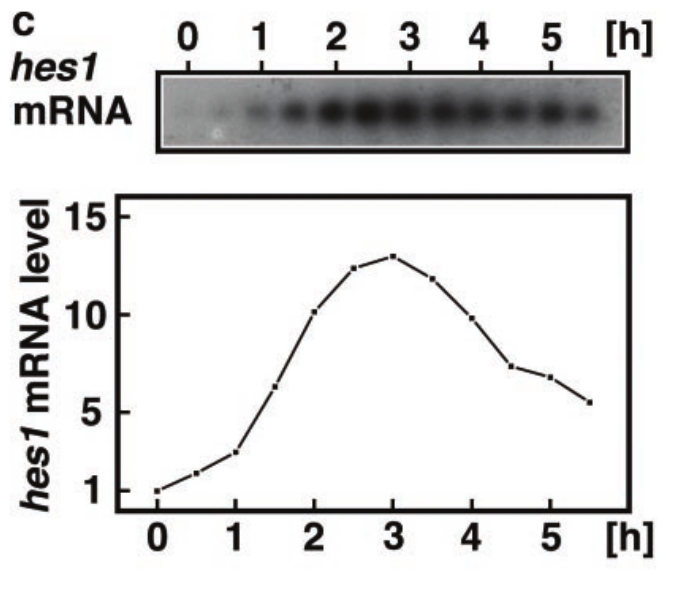}
\end{minipage}
\begin{minipage}{\textwidth}
  \centering
  \includegraphics[width=\linewidth]{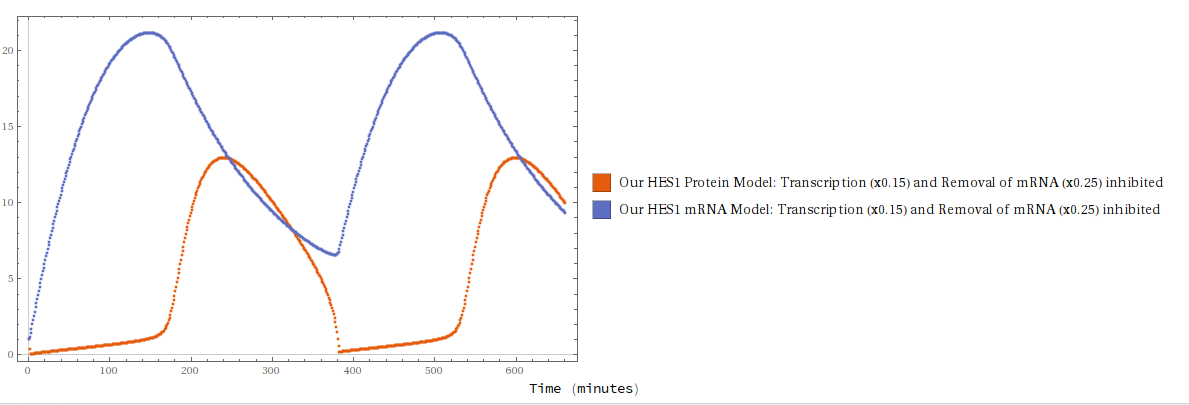}
\end{minipage}
\caption{Above: Hirata's empirical data obtained by infecting C3H10T1/2 cells with cyclohexamide, known to suppress gene transcription and mRNA destruction. Below: Modelling the cyclohexamide intervention in our model by directly manipulating the rate constants of the two affected processes. We're off by smaller factors this time, and we couldn't catch the negative second derivative suggested by the data, but the shapes look about right again.}
\end{figure}

\clearpage

\section{Conclusion}

We have displayed a framework that constructs dynamical systems from petri nets, by assuming Law of Mass Action. The differential equations so obtained are descriptive and qualitatively faithful to empirical data, as unknown rate variables can be tuned -- subject to constraints -- to fit empirical data. The underlying petri nets offer a foothold for structured inquiry and interrogation, and as we have shown in our HES1 gene network case study, tuning the rate coefficients of processes in correspondence with laboratory interventions can yield agreement when the basic processes are well understood. Thus we have proposed, to the best of our knowledge, novel foundations for a scientific method for dynamical systems.\\

We are aware of categorical approaches to dealing with the semantics and phenomena of dynamical systems, such as in \cite{baez_categories_2015}, \cite{baez_biochemical_2018} \cite{baez_biochemical_2018}, \cite{schultz_dynamical_2019} and \cite{baez_compositional_2017}. We have obtained our results independently of these frameworks, as we are unconcerned with the compositionality of these systems, and we have focused our attention on the highly restricted cases of dynamical systems where Law of Mass Action applies. We observe that any compositional framework that is suitable for petri nets with rates may be lifted through our method to induce a compositional structure on dynamical systems specified by first order ODEs in rational functions.\\

The notion of polynomial-algebraic models for dynamic systems is not new. In the field of theoretical computer science, \cite{berry_chemical_nodate} and \cite{abramsky_computational_1993} have proposed models to Linear Logic in terms of ``Chemical Abstract Machines", the syntax of which bears some similarity to that which we have introduced, though there is no discussion of relations to dynamical systems \emph{per se}, as beyond the syntax, the remainder of the formalism consists of proof rules, though inspired by chemical bonding and dissociation, are only applied in the context of proof-theoretic concerns.\\

\cite{veliz-cuba_polynomial_2010} develop a polynomial-algebraic framework for specifying discrete dynamical systems, which they have also explicitly related to logical models and petri-net models. However, they only consider finite fields, which they leverage using a theorem that nonconstructively guarantees the ability to model any function between such finite fields. By contrast, our approach carries no such restriction in model domains, and explicitly requires structural knowledge of the processes that comprise a system: this is the requirement for structured and scientific inquiry.\\

The requirement of foreknowledge of constituent processes hints at a broader theory of composing such processes, which was rediscovered many times, elaborated for instance by \cite{brown_categorical_1995} in a purely formal manner, demonstrating that petri nets may be considered atomic terms in Linear Logic, and hence may be composed with the same Linear Logical structure. From the other end, \cite{engberg_petri_nodate} demonstrated that petri nets are suitable models of Linear Logic in and of themselves. Future work would be to reconcile these two views with a logic of dynamical systems as we have described. The closest extant candidate to such a reconciliation might be \cite{kahramanog_deductive_nodate}, who has devised a deductive compositional approach borrowing ideas from the logical subfield of deep inference, with an explicit view towards constructing a syntax to express and compose dynamical models for complex biochemical systems.\\

Our method brings to mind several immediate use cases. The first is automated hypothesis generation. If one possesses partial information concerning the underlying processes some phenomenon, one can guess at what the missing processes are, and choose promising candidates among those guesses for further investigation based on how well the resulting dynamical systems agree with empirical data. In other words, candidate petri nets with assumptions that yield good dynamical systems may posit hidden constituent processes, and these hidden constituent processes merit discovery or dismissal. Though there may be guesswork involved in choosing those hidden processes, the guesswork is in the space of petri nets, which is relatively structured and constrained, when compared to the space of all possible dynamical systems. Furthermore, so long as known processes are expressed in the petri net along with their tuned rates, any candidate petri net that contains what one already knows as a subnet is forced to agree with all previously established knowledge.\\

If the dynamical models obtained agree with empirical data under interventions targeting every transition in the petri net, then the model is as correct as reasonably possible. Under such circumstances, one may predict the effect of previously unseen linear combinations of known interventions with reasonable confidence. If these predictions are wrong, then either one of the constituent processes is not `atomic', possibly composed of hidden transitions and hidden species (which provokes investigation), or the Law of Mass Action does not hold.\\

The Law of Mass Action is a massively simplifying assumption, that requires large numbers of mindless species interacting within a fundamentally non-spatial arena in random ways. Further, we have only considered ordinary first order differential equations, which amounts to the tacit physical assumption that the processes in the system are totally determined by the participating species: the species are the only numeraire. This rules out easy expression of higher order differential equations. In time we hope to do away with some of these simplifying assumptions.

\clearpage

\begin{appendix}

\begin{algorithm}
\caption{Obtain dynamical system from set of basic processes $\mathbf{B}$ over state variables $\mathbf{S}$.}
\begin{algorithmic}
\STATE{(\text{Initialise variables to hold analytic expressions for each }$\Delta A$)}
\FOR{$A \in S$}
\STATE $\Delta A \leftarrow 0$
\ENDFOR
\STATE{(\text{Initialise variables to hold free rate coefficients for each basic process})}
\FOR{$b_i \in \mathbf{B}$}
\STATE $Rate_i \leftarrow \mathbf{r}_i$
\ENDFOR
\STATE{(\text{For each basic process...})}
\FOR{$b_i \in \mathbf{B}$}
\REQUIRE $b_i \equiv \bigotimes\limits_{1}^{k} n_i A_i \multimap \bigotimes\limits_{k+1}^{k+l}n_jA_j$
\STATE{(\text{...initialise variables to hold onto Law of Mass Action rate and Net change in species...})}
\STATE $LMA \leftarrow 1$
\STATE $Net \leftarrow 0$
\STATE{(\text{...and for each occurrence of a species with multiplicity...})}
\FOR{$i \leq k+l$}
\STATE (\text{...if the occurrence is as an input...})
\IF{$i \leq k$}
\STATE (\text{...account for Law of Mass Action multiplicatively...})
\STATE $LMA \leftarrow LMA \times A_i^{n_i}$
\STATE (\text{...and subtractively account for Net change.})
\STATE $Net \leftarrow Net - n_i A_i$
\STATE (\text{If the occurrence is as an output, additively update Net change.})
\ELSE[$i > k$]
\STATE $Net \leftarrow Net + n_i A_i$
\ENDIF
\STATE (\text{Update the rate of this basic process with LMA})
\STATE $Rate_i \leftarrow Rate_i \times LMA$
\STATE (\text{Split the Net change of this process among all species, with rate.})
\REQUIRE $Net \equiv \sum\limits_{1}^{m} n_m A_m$
\FOR{$A_m \in Net$}
\IF{$n_m < 0$}
\STATE $\Delta A_m \leftarrow \Delta A_m - Rate_i * n_m * A_m$
\ELSE[$n_m > 0$]
\STATE $\Delta A_m \leftarrow \Delta A_m + Rate_i * n_m * A_m$
\ENDIF
\ENDFOR
\ENDFOR
\ENDFOR
\RETURN $\Delta A: (A \in \mathbf{S})$
\STATE $(\text{Output is a dynamical system over } \mathbf{S} \text{ with free variables }\mathbf{r}_i)$
\end{algorithmic}
\end{algorithm}

\end{appendix}

\clearpage

\nocite{*}
\bibliography{sm4dc}
\bibliographystyle{alpha}

\end{document}